\documentclass[reqno, 10pt]{amsart}
\usepackage{amsmath, amsthm, amssymb, graphicx, mathrsfs, enumerate, mathtools}
\usepackage{setspace}
\usepackage{hyperref}
\usepackage{multirow}
\usepackage{mathtools}
\usepackage{tabulary}
\usepackage{etoolbox}
\newcolumntype{K}[1]{>{\centering\arraybackslash}p{#1}}
\usepackage[T1]{fontenc}
\usepackage{ytableau}
\usepackage{vmargin}
\newcommand{\A}{\mbox{$\mathbb A$}}
\newcommand{\N}{\mbox{$\mathcal N$}} 
\def\a{{\alpha}}
\def\b{{\beta}}
\def\d{{\delta}}
\def\g{{\gamma}}
\def\ld{{\lambda}}
\def\v{{\nu}}
\def\z{{\zeta}}
\def\w{{\omega}}
\def\c{{\mathrm{\textbf{c}}}}
\def\X{{\mathrm{X}}}
\def\Y{{\mathrm{Y}}}

\def\S{\mathcal {S}}

\def\AA{{\mathcal A}}
\def\P{{\mathcal P}}
\def\I{{\mathcal I}}

\def\H{{\mathcal H}}

\def\c{{\bf{c}}}
\def\Fqm{{\mathbb F}_{q}^{m}}
\def\Fqn{{\mathbb F}_{q}^{n}}

\newcommand{\Fq}{\mbox{$\mathbb{F}_{q}$}}

\newcommand{\F}{\mbox{$\mathbb F$}}
\newcommand{\Ev}{\mathrm{Ev}}

\newcommand{\RM}{\mathrm{RM}}
\newcommand{\RS}{\mathrm{RS}}
\newcommand{\AC}{\mathrm{AC}}

\newcommand{\GL}{\mathrm{GL}}
\newcommand{\AGL}{\mathrm{AGL}}
\def\Aff{{\normalfont\text{Aff}}}
\newcommand{\G}{\mathrm{G}}
\newcommand{\M}{\mathrm{M}}
\newcommand{\orb}{\mathrm{orb}}
\newcommand{\stab}{\mathrm{stab}}

\newcommand{\GA}{\mathcal{G}_{\mathcal{A}}}

\newcommand{\GG}{\mathcal{G}}
\newcommand{\ZZ}{\mathrm{Z}}

\newcommand{\supp}{\mathrm{Supp}}
\newcommand{\stirling}[2]{\genfrac{[}{]}{0pt}{}{#1}{#2}}

\newtheorem{theorem}{Theorem}[section]
\newtheorem{lemma}[theorem]{Lemma}
\newtheorem{prop}[theorem]{Proposition}

\newtheorem{cor}[theorem]{Corollary}

\theoremstyle{definition}
\newtheorem{defn}[theorem]{Definition}
\theoremstyle{remark}
\newtheorem{remark}[theorem]{Remark}
\numberwithin{equation}{section}

\begin{document}

\title[Enumeration of minimum weight codewords of affine Cartesian codes]{Enumeration of minimum weight codewords of \\ affine Cartesian codes}

\author{Sakshi Dang}
\thanks{Sakshi Dang is supported by Senior Research Fellowship from the Council of Scientific and Industrial Research, Govt. of India, for her doctoral work  at IIT Bombay.}
	\address{Department of Mathematics, 
		Indian Institute of Technology Bombay,\newline \indent
		Powai, Mumbai 400076, India}
		\email{\href{mailto:sakshidang10@gmail.com}{sakshidang10@gmail.com}}

	\author{Sudhir R. Ghorpade}
	\address{Department of Mathematics, 
		Indian Institute of Technology Bombay,\newline \indent
		Powai, Mumbai 400076, India}
\email{\href{mailto:srg@math.iitb.ac.in}{srg@math.iitb.ac.in}}
\thanks{Sudhir Ghorpade is partially supported by  the grant DST/INT/RUS/RSF/P-41/2021 from the Department of Science \& Technology, Govt. of India, and the IRCC award grant RI/0224-10000101-001 from IIT Bombay.}
	
\date{}

\begin{abstract}

Affine Cartesian codes were first discussed by Geil and Thomsen in 2013 in a broader framework and were formally introduced by L\'opez, Renter\'ia-M\'arquez and Villarreal in 2014. These are linear error-correcting codes obtained by evaluating polynomials at points of a Cartesian product of subsets of the given finite field. They can be viewed as a vast generalization of Reed-Muller codes. In 1970, Delsarte, Goethals and MacWilliams gave a 
formula for the minimum weight codewords of Reed-Muller codes. 
Carvalho and Neumann in 2020 considered affine Cartesian codes in a special setting where  the subsets in the Cartesian product are nested subfields of the given finite field, and gave a characterization of their  minimum weight codewords. We use this to give an explicit formula for the number of minimum weight codewords  of affine Cartesian codes in the case of nested subfields. This is seen to unify the known formulas for 
the number of minimum weight codewords of Reed-Solomon codes and Reed-Muller codes.   

\end{abstract}

\maketitle


\section{Introduction}\label{section_intro}

Let $\Fq$ denote the finite field with $q$ elements, where $q$ is the power of some prime $p$. For a positive integer $m$, consider  nonempty subsets $A_1, \ldots,A_m$ of $\Fq$ of size $d_1, \ldots, d_m$ respectively. Without loss of generality, we may assume that
$$
	 1 < d_1  \leq \cdots\leq d_m \leq q .
$$
Define the Cartesian product 
$$
    \mathcal{A}: = A_1 \times \cdots \times A_m \subseteq \F_q^m \quad \text{and} \quad  n := |\mathcal{A}| = d_1 \cdots d_m .
$$
Denote by $\S$, the polynomial ring in $m$ variables $X_1, \ldots, X_m$ over $\Fq$. Let $\X=(X_1, \ldots,X_m)$. Then
$$
    \S: = \Fq[\X] = \Fq[X_1,\ldots,X_m].
$$
For a nonnegative integer $u$, define $\S_{\leq u}(\AA)$ to be the vector subspace of $\S$ having polynomials $f$ such that $\deg f \leq u$ and  $\deg_{X_i} (f) <d_i$ for $1 \leq i \leq m$, i.e., 
$$
    \S_{\leq u}(\AA) := \{f(\X) \in \S:  \deg f \leq u,  \deg_{X_i} (f) <d_i \text{ for } 1 \leq i \leq m\}.
$$
Denote, by convention, $\S_{\leq -1}(\AA) = \{ \bf{0} \}$, the zero subspace of $\S$. Since $|\AA| =n$, we write $\AA = \{P_1, \ldots, P_n\}$. Define the evaluation map  
$$
    \Ev: \S_{\leq u}(\AA) \rightarrow \Fqn \qquad \text{given by} \qquad  f \mapsto \c_f: = (f(P_1), \ldots, f(P_n)).
$$
We note that $\Ev$ is a linear injective map and its image is a linear code, called the {\it affine Cartesian code of order $u$} and is denoted by $\AC_{q}(u, \AA)$. Thus, $\AC_{q}(u, \AA) = \Ev(\S_{\leq u}(\AA))$ is a linear code of length $n = d_1 \cdots d_m$. Moreover, $\Ev$ is a surjective map if $u \geq K : = \sum\limits_{i=1}^{m}(d_i -1)$. 

Affine Cartesian codes were first defined by Geil and Thomsen \cite{Geil} in 2013 in a more general setting where they consider weighted Reed-Muller codes on a Cartesian product of subsets of the given finite field. Independently, these codes were formally introduced by L\'opez, Renter\'ia-M\'arquez and Villarreal in 2014 \cite{Lopez}; in this work, they  
determine the basic parameters of these codes using the theory of complete intersections and Hilbert functions. Among the different evaluation codes that exist in the literature, Reed-Solomon codes and Reed-Muller codes are widely studied and can be seen as special cases of affine Cartesian codes. In fact, 
\begin{itemize}
    \item If m=1, then $\AC_{q}(u, \AA) = \RS_{q}(u+1,d_1)$,  the Reed-Solomon code of length $d_1$ and dimension $u+1$ over $\Fq$.
    
     \item If $ A_1= \cdots = A_m =\Fq$, then $\AA = \Fqm$ and hence $\AC_{q}(u, \AA) = \RM_{q}(u,m)$, the Reed-Muller code of order $u$ in $m$ variables over $\Fq$. 
\end{itemize}

Reed-Muller codes were first introduced by Muller \cite{Muller} in 1954, and a decoding algorithm was later developed by Reed \cite{Reed}. These codes were initially defined over the binary field $\F_2$ and were later extended in 1968 by Kasami, Lin, and Peterson to generalized Reed-Muller codes, which are defined over a finite field $\F_q$. Generalized Reed-Muller codes, denoted by $\RM_q(u, m)$, are obtained by evaluating polynomials of degree at most $u$ in the polynomial ring $\F_q[X_1, \ldots, X_m]$ at all points in the affine space $\A^m(\F_q)$. An extensive literature on binary and $q$-ary Reed-Muller codes can be found in \cite{Assmus_Key, Sloane}. Due to their nice algebraic properties, these codes have been extensively studied and are still an active area of research.  It were the seminal papers by Kasami, Lin, Peterson in 1968 \cite{Kasami} where they determined the basic parameters of these codes and by Delsarte, Goethals, MacWilliams in 1970 \cite{Delsarte_Goethals_MacWilliams} where they gave a characterization of minimum weight codewords in terms of product of linear polynomials that significantly advanced the study of Reed-Muller codes. Other prominent works on Reed-Muller codes include the determination of their automorphism groups by Berger, Charpin \cite{Berger_Charpin}, their generalized Hamming weights by Heijnen, Pellikaan \cite{Heijnen_Pellikaan}, and later by Beelen\cite{Beelen}. A lot of work has also been done towards the determination of higher Hamming weights and next-to-minimum weight codewords of Reed-Muller codes. One can refer to \cite[Section 1]{Carvalho_Neumann} and references within for more details on higher weights of Reed-Muller codes. 

Affine Cartesian codes, being a generalization of Reed-Muller codes have always been under active consideration. Other than the defining works by L\'opez, Renter\'ia-M\'arquez and Villareal   \cite{Lopez}, some major works include determination of their duals and generalized Hamming weights by Beelen, Datta \cite{Beelen_Datta}, their complimentary duals by Lopez, Manganiello, Matthews \cite{Lopez_2} and relative generalized Hamming weights by Datta \cite{Datta}. In a special case, Carvalho and Neumann \cite{Carvalho_Neumann} gave a characterization of minimum weight codewords. They also determined second Hamming weights and  next-to-minimal weight codewords of affine Cartesian codes. One can refer to \cite{Carvalho, Carvalho_Neumann_2, Carvalho_Neumann_3} for more details.

Delsarte, Goethals, MacWilliams in \cite{Delsarte_Goethals_MacWilliams} gave a characterization of minimum weight codewords for Reed-Muller codes. They showed that for $0 \leq  u \leq m(q-1)$, if 
$$
    u=t(q-1)+s
$$
for some uniquely determined integers $t,s$ such that $0 \leq t \leq m$ and $0 \leq s  < q-1$, every minimum weight codeword of the Reed-Muller code $\RM_q(u,m)$ is obtained by evaluation of $\c_f$ of polynomials $f \in \S_{\leq u}(\AA)$, given by
$$
    f(X_1, \ldots, X_m) = \g \prod_{i=1}^{t} (1-(L_i +b_i)^{q-1}) \prod_{j=1}^{s}(L_{t+1}-\w_j),
$$
where  $\g  \in \F_q^*$, $L_1, \ldots, L_{t+1} \in \Fq[X_1, \ldots, X_m]$ are linearly independent linear homogeneous polynomials, $b_1, \ldots, b_t \in \Fq$ and $\Omega:=\{ w_1, \ldots, \w_s\} \subseteq \Fq$. This result showed that all the minimum weight codewords lie on a special type of hyperplane arrangement.  Moreover, they also enumerated the minimum weight codewords and showed that if $\mathrm{N}_{q}(u,m)$ denotes the set of minimum weight codewords of $\RM_q (u,m)$, then 
$$
    \left| \mathrm{N}_{q}(u,m) \right| =
    \begin{cases}
        (q-1)q^{t} \stirling{m}{t}_{q} & {\text{ if }s =0},\\[.5em]    
        (q-1)q^{t} \stirling{m}{t}_{q}  \stirling{m-t}{1}_{q} \binom{q}{s}  & {\text{ if } s \neq 0},\\
    \end{cases}
$$
where $\stirling{m}{t}_{q}$ denotes the $q$-binomial coefficient and is defined to be the number of $t$-dimensional linear subspaces of $\Fqm$. In 2012, Leducq in \cite{Leducq} gave an alternative proof for the above result of Delsarte, Goethals and MacWilliams \cite{Delsarte_Goethals_MacWilliams} using the concept of \textquoteleft reduced'  polynomials. A generalization of this result to affine Cartesian codes has proven to be a difficult problem due to arbitrary nature of subsets of the Cartesian product under consideration. However, some work has recently been done by Geil \cite{Geil_2} in this direction where he showed that not all minimum weight codewords are obtained by polynomials which can be expressed as a product of linear factors. Extending the work of Delsarte  Goethals and MacWilliams \cite{Delsarte_Goethals_MacWilliams}, Carvalho and Neumann  \cite{Carvalho_Neumann} in 2020 considered a special case of affine Cartesian codes and gave a similar characterization for its minimum weight codewords. They assumed subsets $A_1, \ldots, A_m$ of the Cartesian product $\AA$ to be subfields of $\Fq$. Moreover, they further assumed that 
$$
    A_1 \subseteq \cdots \subseteq A_m \subseteq \Fq.
$$
We address the problem of enumerating the minimum weight codewords in this specific context. We write 
$$ 
    \AA = F_1 \times\cdots \times  F_1 \times \cdots\cdots \times F_{\lambda} \times \cdots \times   F_{\lambda} \subseteq \F_q^m,
$$
where $F_1 \subsetneq \cdots \subsetneq F_{\lambda} \subseteq \Fq $ are distinct subfields of $\Fq$ and for $1 \leq t \leq \lambda $, let $F_t$ is repeated $\mu_t$ times in $\AA$. Then, it follows that
$$
    \mu_1 + \cdots + \mu_{\lambda } = m.
$$
Define $s_{0}  =0$, and 
$$
    s_{t} := \mu_1 + \cdots + \mu_t \quad \text{ for } \quad 1 \leq t \leq \lambda. 
$$
Since $\mu_t > 0$ for $1 \leq t \leq \lambda$, we get
$$ 
    0 < s_{1}  < s_{2}  < \cdots<   s_{\lambda}  = m.
$$
Then, it follows from the division algorithm that for every $k$ such that $1 \leq k \leq m$, there exists a unique integer $t =t_k$ such that $s_{t_{k}-1}  < k \leq s_{t_{k}} $. 

Denoting by $\N_q(u, \AA)$, the set of minimum weight codewords of $\AC_{q}(u, \AA)$, we completely determine $\left | \N_{q} (u, \AA) \right|$ for all values of $u$ under consideration and is given by the following theorem. 

\begin{theorem}\label{thm_main}
Let $m \geq 1$ and  $\AA = F_1 \times \cdots \times F_1 \times \cdots \times F_{\lambda} \times \cdots \times F_{\ld}  = F_1^{\mu_1} \times \cdots \times F_{\ld}^{\mu_{\ld}} \subseteq \F_q^m $ as before. Then
$$
    \left| \N_{q}(0, \AA) \right| = q-1.
$$
For $1 \leq u \leq K= \sum\limits_{i=1}^{m} (d_i-1)$, write
$$
    u = \sum_{i=1}^{j}(d_i-1)+\ell,
$$
where $j, \ell$ are uniquely determined integers such that $0 \leq j < m$ and $0 < \ell \leq d_{j+1}-1$. Then
\begin{equation*}
    \left| \N_{q}(u, \AA) \right| = \begin{cases}
            (q-1) \bigg(\prod\limits_{\underset{i \neq k}{i=1}}^{j+1} d_i \bigg)  \displaystyle{\stirling{\mu_r}{j+1-s_{r-1}} _{d_{j+1}}}
            & \text{ if } \ell = d_{j+1}-1,\\
           \left| \N_{q}^{(j+1)}(u, \AA) \right| +  \sum\limits_{t = t_{k_0}}^{r-1} \left| \N_{q}^{(s_{t} )}(u, \AA) \right|   & \text{ if } \ell < d_{j+1}-1,\\
        \end{cases}
    \end{equation*}
    where $r = t_{j+1}$,   $1 \leq k_0 \leq j+1$ is the least integer $k$ such that $d_{k} \geq d_{j+1}-\ell$ and  $|\N_{q}^{(k)}(u, \AA)|$ is given by
    \begin{itemize}
        \item If $t_k = r$, 
        $$
        \left| \N_{q}^{(k)}(u, \AA) \right| = (q-1) \bigg(\prod\limits_{\underset{i\neq k}{i=1}}^{j+1} d_i \bigg) \displaystyle{\stirling{\mu_r}{j-s_{r-1} }_{d_{j+1}}  \stirling{s_{r }-j}{1}_{d_k}} \binom{d_{k}}{d_{j+1}-\ell}.
        $$

        \item If $t_k < r$, 
        $$
         \left| \N_{q}^{(k)}(u, \AA) \right| =   (q-1) \bigg(\prod\limits_{\underset{i \neq k}{i=1}}^{j+1} d_i \bigg) \bigg(\prod\limits_{i=s_{t_k} +1}^{j+1} d_i \bigg) \displaystyle{\stirling{\mu_r}{j+1-s_{r-1} }_{d_{j+1}}  \stirling{\mu_{t_k}}{1}_{d_k}}\binom{d_k}{d_{j+1}-\ell}.  
        $$
    \end{itemize}
\end{theorem}  

This article is organized as follows. In Section \ref{section_preliminaries}, we give some preliminary results pertaining to basics of finite fields, group theory and coding theory that will be used later. In Section \ref{section_special_case}, we consider the special case of affine Cartesian codes where subsets of the Cartesian product $\AA$ are assumed to be nested subfields of $\Fq$. We give a characterization of the affine group with respect to $\AA$ in terms of non-singular matrices. Section \ref{section_enumeration} focuses on the problem of enumerating minimum weight codewords and is further divided into two subsections following different cases that we need to consider. We conclude this article with Section \ref{section_conclusion} where we give a proof of Theorem \ref{thm_main} and deduce some special cases.   

\section{Preliminaries}\label{section_preliminaries}
Let $m$ and $\Fq$ be as defined in Section $\ref{section_intro}$ and remain fixed throughout the article. We begin with an elementary yet crucial property related to elements of finite fields.
 
\begin{lemma}[\cite{Hou}, Lemma 2.21]\label{lemma_a^s}
Let $s$ be an integer and $0 \leq s \leq q-1$. Then \begin{equation} \label{eqn_a^s}
    \sum_{ a \in \mathbb{F}_q} a^{s} = 
    \begin{cases}
        0 & \text{ if  $0 \leq s \leq q-2 $}, \\
        -1 & \text{ if  $s = q-1$}. \\
  \end{cases}
\end{equation} 
\end{lemma} 


Denote by $\GL(m, \Fq)$,  the group of $m \times m $ non-singular matrices over $\Fq$. 

\begin{defn}\label{defn_AGL}
The {\bf affine general linear group} of degree $m$ over $\mathbb{F}_{q}$ is
\begin{equation*}
    \AGL(m, \mathbb{F}_{q}) = \left \{  
    \begin{bmatrix}
    A& b \\
    0 & 1\\
    \end{bmatrix}
    : A \in \GL(m,\mathbb{F}_{q} ), b \in \mathbb{F}_{q}^{m} \right \}.
\end{equation*}
\end{defn}
\noindent
It can also be described as follows.
\begin{equation}\label{eqn_AGL}
    \AGL(m ,\Fq):= \{\sigma_{A,b} : \F_q^m \rightarrow \F_q^m | \sigma_{A,b}(\X) = A \X +b \text{ for some } A \in \GL(m, \Fq), b \in \F_q^m  \}. 
\end{equation}
    
Note that $\AGL(m, \Fq)$ is a group with respect to composition of maps, i.e., 
$$
    \sigma_{A', b'} \circ \sigma_{A,b}(\X) = \sigma_{A', b'}(A \X+b) = (A'A)\X + (A'b+b'),
$$
which is same as matrix multiplication
\begin{equation*}
    \begin{bmatrix}
    A' & b' \\
    0 & 1\\
    \end{bmatrix}
    \begin{bmatrix}
    A & b \\
    0 & 1\\
    \end{bmatrix} = \begin{bmatrix}
    A'A & A'b +b' \\
    0 & 1\\
    \end{bmatrix}.
\end{equation*}
Observe that $|\AGL(m, \Fq)| = \left|\Fqm \right| \cdot \left| \GL(m,\Fq) \right| = q^m \prod\limits_{i=0}^{m-1}(q^m-q^i)$.

\begin{defn}\label{defn_group-action}
Let $(G,*)$ be a group with identity $e$ and let $X$ be a set. A {\bf group action} $\Phi$ of $G$ on $X$ is a map 
$$
    \Phi: G \times X \rightarrow X, 
$$
such that  for all $g_1, g_2 \in G, x \in X$, it satisfies 
\begin{enumerate}
    \item $\Phi(e,x) =x$. 
    \item $\Phi(g_1*g_2, x) = \Phi(g_1, \Phi(g_2,x))$.
\end{enumerate}
\end{defn}
\noindent
One can refer to \cite[Section 1.7]{Dummit_Foote} and \cite[Chapter 4]{Dummit_Foote} for more details on group actions.  
\begin{defn}
The {\bf orbit} and {\bf stabilizer} of an element $x \in X$ w.r.t. the action $\Phi$ of $G$ on $X$  are defined respectively by
$$
    \orb_{G}(x) : = \{\Phi(g,x): g\in G\}
    \quad \text{ and }\quad
    \stab_{G}(x):= \{ g \in G: \Phi(g,x) =x\}.
$$
\end{defn}

The orbit and stabilizer of an element are closely related and this relation is given by the following result, called the {\it orbit-stabilizer theorem}. 

\begin{theorem}[\cite{Dummit_Foote}, Chapter 4, Proposition 2] \label{thm_orb-stab}
Let $\Phi$ be an action of a group $G$ on a finite set $X$. Then for any $x \in X$, 
$$
     |\orb_{G}(x)| = \frac{|G|}{ |\stab_{G}(x)|}.
$$
\end{theorem}

Let $0 \leq s \leq q-1$. Denote by $\P_{s}(\Fq)$ the collection of all subsets of $\Fq$ having exactly $s$ elements, i.e.,
$$
    \P_{s}(\Fq):= \{\Omega \subseteq \Fq: |\Omega| =s \}.
$$
Therefore, 
$$
    |\P_{s}(\Fq)| =\binom{q}{s}.
$$
Consider the group action $\Phi: \AGL(1, \Fq) \times \P_{s}(\Fq) \rightarrow \P_{s}(\Fq)$ given by 
$$
    \Phi( \sigma_{a,b}, \Omega) = a\Omega +b: = \{ a \omega_1+b, \ldots, a \omega_s+b\} \text{ for } \Omega = \{ \omega_1, \ldots, \omega_s\}.
$$
Note that $\Phi$ is a well-defined group action. Denote by $[\Omega]$, the orbit of $\Omega $ under the group action $\Phi$, i.e.,
$$
    [\Omega]:= \{\Phi( \sigma_{a,b}, \Omega): \sigma_{a,b} \in \AGL(1, \Fq)\}.
$$ 
We know that any two orbits are either disjoint or identical. Since $\P_{s}(\Fq)$ is a finite set, there will be finitely many disjoint orbits. Suppose there are $\eta$ disjoint orbits and let $\Omega_1, \ldots, \Omega_{\eta} $ be their representatives. Then, it must follow that
$$
    |\P_{s}(\Fq)| = \sum_{i=1}^{\eta} \left| \right [\Omega_i]|. 
$$
But, Theorem \ref{thm_orb-stab} tells us that for $\Omega \in \P_{s}(\Fq)$, 
$$
    \left|[\Omega]\right| = \frac{|\AGL(1, \Fq)|} {|\Delta_{\Omega}|},
$$
where 
\begin{equation}\label{eqn_stab1}
    \Delta_{\Omega} := \{\sigma_{a,b} \in \AGL(1, \Fq): a\Omega+b = \Omega  \},
\end{equation}
is the stabilizer of $\Omega$ with respect to the group action $\Phi.$ Therefore, 
$$
    |\P_{s}(\Fq)| = \sum_{i=1}^{\eta} \frac{|\AGL(1, \Fq)|}{|\Delta_{\Omega_i}|},  
$$
which implies that
\begin{equation}\label{eqn_stab-reciprocal-sum}
      \sum_{i=1}^{\eta}\frac{1}{|\Delta_{\Omega_i}|} = \frac{1}{q(q-1)} \binom{q}{s}.
\end{equation}

\subsection{Affine Cartesian codes}

Let $A_1, \ldots, A_m$ be nonempty subsets of $\Fq$ and let $|A_i| =d_i$. For $1 \leq i \leq m$, we write $A_i = \{ \gamma_{i,1}, \gamma_{i,2},\ldots,\gamma_{i, d_i} \}$. Without loss of generality, we may assume that
$$ 
    1 < d_1 \leq d_2 \leq \cdots\leq d_m \leq q.
$$
Define the cartesian product 
$$
    \mathcal{A} = A_1 \times A_2 \times \cdots \times A_m \subseteq \Fqm.
$$
Recall, from Section \ref{section_intro},
$$
    n = \prod_{i=1}^{m}d_i  \quad \text{ and } \quad K = \sum\limits_{i=1}^{m} (d_{i} -1).
$$
Then $ |\AA| = n$ and we write $\AA = \{P_1, \ldots, P_n\}$. Define 
$$
    \mathcal{I}(\AA) := \{ f \in \S: f(P)=0 \text{ for all } P \in \AA \}.  
$$
Note that $\I(\AA)$ is an ideal of $\S$, called {\it the vanishing ideal of $\AA$}. Moreover, $\I(\AA)$ is given by  
$$
    \I(\AA) = \left\langle g_1, \ldots, g_m \right\rangle,
$$
where the generators $g_1, \ldots, g_m$ of $\I(\AA)$ are univariate polynomials defined by 
\begin{equation}\label{eqn_def-g_i}
    g_i(X_i) = \prod_{j=1}^{d_i} (X_i -\gamma_{i,j}) \text{ for } 1 \leq i \leq m.
\end{equation}

We know that for a subset $Y \subseteq \Fqm$, the function $f: Y \rightarrow \Fq$ can always be expressed by some polynomial $\mathrm{p} \in \S$. Let $\S(\AA)$ be the $\Fq$-algebra of functions defined on $\AA$. Consider the map 
$$
    \psi: \frac{\Fq[X_1, \ldots, X_m]}{\mathcal{I}(\AA)} \rightarrow \S(\AA) \text{ given by } f \mapsto f\pmod{ \I(\AA) }.
$$
Note that $\psi$ is an isomorphism and hence, for every $f \in \S(\AA)$, there exists a unique polynomial $\mathrm{p} \in \S$ such that $\deg_{X_i}(\mathrm{p}) < d_i$ for $1 \leq i \leq m$ and $\psi(\mathrm{p}+\I(\AA))=f$. We say that $\mathrm{p}$ is the {\it reduced} polynomial associated to $f$ and define the degree of $f$ as being the degree of $\mathrm{p}$. Therefore, without loss of generality, we will identify elements of $\S(\AA)$ with their reduced polynomials, i.e., 
$$
    \S(\AA) = \frac{\S}{\mathcal{I}(\AA)} = \{ f \in \S: \deg_{X_i}(f) < d_i \}.
$$
For any integer $u \geq 0$, denote by $\S_{ \leq u}(\mathcal{A})$, the vector subspace of $\S(\AA)$ consisting of polynomials $f$ with $\deg(f) \leq u$ and $\deg_{X_i}(f) < d_i$ for $1 \leq i \leq m$. Then, we can write
\begin{equation}\label{eqn_S_u}
    \S_{\leq  u}(\mathcal{A}) = \{f \in \Fq[X_1,\ldots,X_m]: \deg(f) \leq u, \deg_{X_i}(f) < d_i \text { for } i=1,\ldots,m \}. 
\end{equation}
By convention, we denote $\S_{-1}(\AA) =\{0\}$.

\begin{defn}
For $f \in \S(\AA)$, we define the set 
$$
    \supp_{\AA}(f) := \{ P \in \AA: f(P) \neq 0 \},
$$
and is called the \textbf{$\AA$-support} of a function $f$. Further, cardinality of the set $\supp_{\AA}(f)$ is called $\textbf{Hamming weight}$ of $f$ and is denoted by $|f|_{\AA}$. 
\end{defn} 

\begin{defn}
We denote the \textbf{zero set} of $f \in S(\AA)$ by $\ZZ_{\AA}(f)$ and is defined by 
\begin{equation*}
    \ZZ_{\AA}(f) := \{ P \in \AA: f(P) = 0\} = \AA \setminus \supp_{\AA}(f).
\end{equation*}
Then, $|\ZZ_{\AA}(f)| = |\AA|- |\supp_{\AA}(f)| = n-|f|_{\AA}$.
\end{defn}

\begin{defn}\label{defn_ACC}
For $u \geq 0$, the {\bf evaluation map $\Ev$} is defined as follows: 
\begin{equation}\label{eqn_Ev}
     \Ev: \S_{\leq u} (\mathcal{A}) \rightarrow \mathbb{F}_{q}^{n} \text{ given by } f \mapsto \c_f \text{ where } \c_f := (f(P_1), f(P_2),\ldots,f(P_n)).
\end{equation}
Then, $\Ev$ is a linear injective map and its image is a linear subspace of $\mathbb{F}_{q}^{n}$. Thus, for $u \geq 0 $,  $\AC_{q}(u, \mathcal{A}) = \Ev(\S_{\leq u}(\mathcal{A}))$ is a linear code of length $n = d_1 d_2 \cdots d_m$ and is called the {\bf affine Cartesian code of order $u$} over $\Fq$. Furthermore, note that the map $\Ev$ is surjective if $u \geq K$. By convention, we denote by $\AC_{q}(-1, \AA) = \{ 0\}$, the zero subspace of $\F_q^n$.
 \end{defn}
\begin{remark}
If m=1, then $\AC_{q}(u, \AA) = \RS_{q}(u+1,n)$,  the Reed-Solomon code of 	length $n$ and dimension $u+1$ over $\Fq$.
\end{remark}

\begin{remark} 
Reed-Muller code, $\RM_{q}(u,m)$ is a particular case of affine Cartesian code $
\AC_{q}(u, \AA)$ when $A_1= \cdots = A_m =\Fq$, i.e., 
$$	
	\AA = \F_q^m \text{ implies } \RM_{q}(u,m) = \AC_{q}(u, \AA) .
$$
\end{remark}
The next theorem gives the basic parameters of affine Cartesian codes.

\begin{theorem}  \label{thm_parameters_AC}
For $0 \leq u \leq K$, the affine Cartesian code $\AC_{q}(u, \AA)$ has the following parameters:
\begin{itemize}
    \item length, n :
    $$
        n =d_1 \cdots d_m.
    $$
    
    \item dimension, $\rho_q(u, \AA)$ [\cite{Lopez}, Theorem 3.1]:
        \begin{multline*}
            \rho_q(u, \AA)= \binom{m+u}{u} - \sum_{1 \leq i \leq m} \binom{m+u - d_{i}}{u - d_{i}}  + \sum_{ i < j } \binom{m+u - (d_i + d_j)}{ u- (d_i + d_j)} \\
            - \sum_{i<j<k} \binom{m+u - (d_i + d_j+d_k)}{ u- (d_i + d_j+d_k)} + \ldots+ (-1)^{m} \binom{m+u - (d_1 +\ldots+ d_m)}{ u- (d_1+\ldots + d_m)}. 
        \end{multline*}

    \item minimum distance, $\d_{q}(u, \AA)$ [\cite{Lopez}, Theorem 3.8]: 
    $$
        \delta_{q}(0,\AA) =n,
    $$
    and for $1 \leq u \leq K$, 
    $$
        \delta_{q}(u,\AA) = (d_{j+1} - \ell) d_{j+2} \cdots d_m,
    $$
    where $ 0 \leq j < m$ and $0 < \ell \leq d_{j+1} -1$ are uniquely determined integers such that $u =  \sum\limits_{i=1}^{j}(d_i -1) + \ell$. 
    
    \item dual, $\AC_{q}(u, \mathcal{A})^{\perp}$ [\cite{Beelen_Datta}, Theorem 5.7]: 
    $$
        \AC_{q}(u, \mathcal{A})^{\perp} = \{ (w_1 c_1 , w_2 c_2 , \ldots, w_n c_n): (c_1,\ldots,c_n) \in \AC_{q}(K-u-1, \mathcal{A})\},
    $$
    where $\frac{1}{w_j} = \prod\limits_{i=1}^{m} g_{i}'(P_j)$ for $1 \leq j \leq n$ and $g_i's$ are given by equation \eqref{eqn_def-g_i}.
\end{itemize}
\end{theorem}


\section{Affine Cartesian codes for Nested Subfields}\label{section_special_case}

We consider a special family of affine Cartesian codes in which the subsets $A_1, \ldots, A_m$ are subfields of $\Fq$. Furthermore, suppose 
$$
    A_1 \subseteq \cdots \subseteq A_m \subseteq \Fq.
$$
\noindent
Then, we can write 
\begin{equation} \label{eqn_AA}
    \AA = F_1 \times \cdots \times F_1 \times \cdots \cdots \times F_{\ld} \times \cdots \times F_{\ld},
\end{equation}
where $F_1 \subsetneq \cdots \subsetneq F_{\ld} \subseteq \Fq $ are distinct subfields of $\Fq$. For $1 \leq t \leq  \lambda$, let $F_t$ is repeated $\mu_t$ times in $\AA$. Then 
$$
    \sum_{t=1}^{\ld} \mu_t=m.
$$
Let $|F_t| = p^{r_t}$, then $r_1 \mid r_2 \mid \cdots \mid r_{\ld}$. Therefore, 
$$
    n  = |\AA| = \prod_{t=1}^{\ld} (p^{r_t})^{\mu_t} = p^{\left(\sum\limits_{t=1}^{\ld} r_t  \mu_t \right)}
$$
Let $\a_t$ be a primitive element of $F_t$ for $ 1 \leq t \leq \lambda$. Then, $F_{t}^{*} = \langle \a_t \rangle $. Define $s_0  =0$ and for $1 \leq t \leq \lambda$,
$$
    s_{t}  := \mu_1 +\ldots +\mu_t = s_{t-1}  + \mu_t.
$$
Note that $s_{\lambda}  =m$. Recall, for each $1 \leq k \leq m$, there exists a unique integer $t =t_k \in \{1, \ldots, \lambda\}$ such that $s_{t_{k}-1}  < k \leq s_{t_k} $. If $\AA = A_1 \times \cdots \times A_m$ with $|A_i| = d_i$. Then, upon comparing with equation $\eqref{eqn_AA}$, we get
$$
    A_i = F_{t_i} \text{ for } s_{{t_i}-1}  < i \leq s_{t_i} .
$$
and hence 
$$
    d_i = p^{r_{t_i}} \text{ whenever } s_{{t_i}-1}  < i \leq s_{t_i} .
$$
In this case, the ideal $\I(\AA) = \{f \in \S(\AA): f(P) =0 \text{ for all } P \in \AA \}$ can be rewritten as
\begin{equation}\label{I_A}
    \I(\AA) = \langle  X_1^{d_1} -X_1, \ldots, X_{m}^{d_m} -X_{m}\rangle.
\end{equation} 

Recall, the affine General Linear group $\AGL(m, \Fq)$ is the group of transformations of $\F_q^m$ of the form $\X \mapsto A\X+b$ for some $A\in \GL(m, \Fq)$ and $b \in \F_q^m$. Carvalho and Neumann \cite{Carvalho} characterized the minimum weight codewords of $\AC_{q}(u, \AA)$ when $\AA$ is a Cartesian product of nested subfields of $\Fq$. Before their main result, they define {\it the affine group associated to $\AA$}, which is given below. 
\begin{defn}\label{defn_Aff_AA}
    The {\bf affine group associated to $\AA$}, denoted by $\Aff(\AA)$ is given by
    \begin{equation}\label{Aff_A}
        \Aff(\AA) = \{ \phi : \AA \rightarrow \AA \text{ given by } \phi = \sigma |_{\AA} \text{ for some } \sigma \in \AGL(m, \Fq) \text{ and } \phi(\AA) =\AA \}.
    \end{equation}
    Further, they say that  polynomials $f, g \in \S(\AA)$ are said to be {\bf $\AA-$equivalent} if there exists $\phi \in \Aff(\AA)$ such that $f =g \circ \phi$.
\end{defn}

\begin{theorem}[\cite{Carvalho}, Theorem 3.5] \label{thm_carvalho}
For $1 \leq  u \leq K$, let 
$$
    u=\sum\limits_{i=1}^{j}(d_i-1) +\ell, \text{ for some } 0\leq j < m \text{ and }  0 < \ell \leq d_{j+1}-1.
$$
The minimum weight codewords of $\AC_{q}(u,\AA)$ are given by evaluations of $f \in \S_{ \leq u}(\AA)$ where $f$ is $\AA$-equivalent to the functions of the form
$$
    h_{k}^{\Omega}(\X) =\g \prod_{\underset{i\neq k}{i=1}}^{j+1}(1-X_{i}^{d_i-1})\prod_{\a=1}^{d_{k}-(d_{j+1}-\ell)}(X_k-\w_{\a}),
$$
for some $ \g \in \F_q^{*}$, $1 \leq k \leq  j+1 $ such that $ d_k \geq d_{j+1}-\ell$, and $\Omega:= \{\w_1, \ldots, \w_{d_{k}-(d_{j+1}-\ell)} \} \subseteq~A_k$. 
\end{theorem}

\subsection{A Characterization of the Affine Group}
In this section, we characterize the affine group w.r.t $\AA$ given in definition $\ref{defn_Aff_AA}$ in terms of nonsingular matrices and hence determine the size of this affine group. 

\begin{defn}\label{def_G}
We define a subset $\G \subseteq \GL(m, \Fq)$ by
\begin{multline}\label{G}
    \G: =  \{A \in \GL(m, \Fq): A=(a_{i,j})_{m \times m} \text{ given by } a_{i,j} \in F_t 
    \text{ whenever } s_{{t_i}-1}  < i \leq s_{t_i} , \\ 
    \text{ and } 
    a_{i,j} =0 \text{ for } j>s_{t_i} , \text{ for } 1 \leq t_i \leq \lambda \}.
\end{multline}
\end{defn}
Note that elements of $\G$ are block lower triangular matrices. Thus, an element $A \in \G$ can be written as 
$$
    A =\begin{bmatrix}
    A_{11} & \bf{0} & \cdots & \bf{0} \\
    A_{21} & A_{22} & \cdots & \bf{0} \\
    \vdots &  & \ddots & \vdots\\
    A_{\ld 1 } & A_{\ld 2} & \cdots & A_{\ld \ld} \\
    \end{bmatrix},
$$
where $A_{ij} \in M_{\mu_i \times \mu_j}(F_i)$ and $A_{ii} \in \GL(\mu_i, F_i)$ for $i =1, \ldots, \ld$. It is easy to observe that $\G$ is a subgroup of $\GL(m, \Fq)$. The next result gives a characterization of the affine group associated to $\AA$ in terms of nonsingular matrices. 

\begin{prop}\label{prop_affine  group}
The affine group associated to $\AA$ is given by 
    $$
        \Aff(\AA) = \big\{ \phi: \AA \rightarrow \AA \text{ given by } \hspace{0.05 in} \phi = \sigma_{A,b}|_{\AA} \normalfont\text{ where } \sigma_{A,b} \in \AGL(m, \Fq) \text{ for } A \in \G, b\in \AA \big \}.
    $$
\end{prop}

\begin{proof}
Let the right hand side be denoted by $\H_{\AA}$, i.e.,
\begin{equation}\label{Aff_H}
    \H_{\AA} := \big\{ \phi: \AA \rightarrow \AA \big| \hspace{0.05 in} \phi = \sigma_{A,b}|_{\AA} \normalfont\text{ where } \sigma_{A,b} \in \AGL(m, \Fq), \text{ for some } A \in \G, b\in \AA \big \}.
\end{equation}
Without loss of generality, we consider elements of $\AA$ to be column vectors. We first show that $\H_{\AA} \subseteq \Aff(\AA)$. To see this, let $\phi \in \H_{\AA}$. Then $\phi : \AA \rightarrow \F_q^{m}$ is given by $\phi (\X) = A \X + b$ for some $A \in \G$, $b \in \AA$. We need to show that $\phi$ is a bijection. Since $\AA$ is a finite set, it suffices to show that $\phi$ is injective and $\phi(\AA) \subseteq \AA$. Let $\X = (X_1, \ldots, X_m) \in \AA$. It implies $X_i \in F_{t_i}$ for $s_{{t_i}-1} < i \leq s_{t_i} $ and for $1 \leq t_i \leq \lambda $. For $1 \leq i \leq m$, let $\X_i$ denotes the $i^{th}$ component of $\X= (X_1, \ldots, X_m)$. Then,  
    \begin{equation*}
            (\phi(\X))_i = (A \X+b)_i = \sum_{j=1}^{m} a_{i,j} X_j + b_i.
    \end{equation*}
Since $A \in \G$, we have $a_{i,j} = 0$ for $j >s_{t_i} $ for $ 1 \leq t_i \leq \lambda$. Therefore,
$$
    (\phi(\X))_i = \sum\limits_{j=1}^{s_{t_i} } a_{i,j} X_j +b_i \in F_{t_i}.
$$ 
Hence, $\phi(\X) \in \AA$. Now, we need to show that $\phi$ is injective. Let $\phi(\X) = \phi(\Y)$ for some $\X, \Y \in \AA$. Therefore, it follows that $A \X +b = A \Y +b$ which implies $A(\X-\Y) =0$, but $A \in \GL(m, \Fq)$ and hence $\X= \Y$. Thus, $\phi : \AA \rightarrow \AA$ is a bijection and therefore, $\phi \in \Aff(\AA)$.

On the other hand, to show $\H_{\AA} \subseteq \Aff(\AA)$,  let $\phi \in \Aff(\AA)$, i.e., $\phi : \AA \rightarrow \AA$ is given by $\phi(\X) = A \X+b$ where $A \in \GL(m, \Fq)$ and $b \in \F_q^{m}$. We need to show that $\phi \in \H_{\AA}$. Let $C_1, \ldots, C_m$ denotes the columns of matrix $A$. We consider two cases:

 \medskip
 
\noindent
{\bf Case 1:}
$\underline{b=0}:$\\
    In this case, $\phi: \AA \rightarrow \AA$ is given by $\phi(\X) =A\X$. Let $e_1, \ldots, e_m$ denote the standard basis vectors of $\Fqm$. Clearly, $e_i \in \AA$ for all $1 \leq i \leq m$, and hence $\phi(e_i) =C_i \in \AA $ for all $i \in \{1, \ldots,m \}$. Thus, every column of matrix $A$ is an element of $\AA$. It follows that for $1 \leq t_i \leq \lambda$ such that $s_{{t_i}-1}  < i\leq s_{t_i} $, we get $a_{i,j} \in F_{t_i}$ for all $ 1\leq j \leq m$. Now, it remains to show that for $1 \leq t_i \leq \lambda $, $a_{i,j} =0 $ for $j > s_{t_i} $. We prove the result by using induction on $\lambda$. If $\lambda =1$, there is nothing to prove. Let $\lambda =2$. Then $s_2  = \mu_1 + \mu_2 =m$. Let 
    $$
        A = \begin{bmatrix}
        A_{11}& A_{12}\\
        A_{21}& A_{22}\\
    \end{bmatrix} \in \GL(m, \Fq),
    $$
    be a block matrix where $A_{ij}$ is a matrix of size $\mu_i \times \mu_j$ with entries from $F_i$ for $i, j \in \{1, 2 \}$. We need to show $A_{12} =0$, i.e., $a_{i,j}=0$ for $1 \leq j \leq \mu_1$, $j > \mu_1$. Let $\X^{(1)}=(\underbrace{\a_1, \ldots, \a_1}_{\text{m-1 times}}, \a_2)^{T} \in \AA$. Then $\phi(\X^{(1)}) =A \X^{(1)} \in \AA$. But, for $1 \leq i \leq m$, we can write
    $$
        (A \X^{(1)})_i = \sum_{j=1}^{m}a_{i,j} \X^{(1)}_j = \sum_{j=1}^{m-1} a_{i,j} \a_1 + a_{i,m} \a_2.
    $$
    If $1 \leq i \leq \mu_1$, then $\sum\limits_{j=1}^{m-1} a_{i,j} \a_1 + a_{i,m} \a_2 \in F_1$, but $a_{i,j}, \a_1 \in F_1$ which implies that $a_{i,m} \a_2 \in F_1$. If $a_{i,m } \neq 0$ for some $i$, then $\a_2 \in F_1$, which is a contradiction. Hence, $a_{i,m} = 0$ for $1 \leq i \leq \mu_1$. Now let $\X^{(2)} = (\underbrace{\a_1, \ldots, \a_1}_{\text{m-2 times}},\a_2, \a_2)^{T} \in \AA$. Then for $1 \leq i \leq \mu_1$, $(A \X^{(2)})_i = \sum\limits_{j=1}^{m-2} a_{i,j} \a_1 + a_{i,(m-1)} \a_2 \in F_1$. Thus, $a_{i,(m-1)} \a_2 \in F_1$, which is true if and only if $a_{i,(m-1)} =0$ for $1 \leq i \leq \mu_1$. Thus, proceeding inductively, we get $A_{12}=\bf{0}$.

    Now let us assume the result holds for $\lambda-1$.
    Consider $\X^{(1)} = (\underbrace{\a_1, \ldots, \a_1}_{\text{m-1 times}}, \a_{\ld})^{T} \in \AA$. Then $\phi(\X^{(1)})= A \X^{(1)} \in \AA$. Therefore, 
    $(A \X^{(1)})_i \in F_{\ld-1}$ for $s_{\ld-2}  < i \leq s_{\ld-1} $. It implies $\sum\limits_{j=1}^{m-1} a_{i,j} \a_1 + a_{i,m} \a_{\ld} \in F_{\ld-1}$. But, for $s_{{t_i}-1} < i \leq s_{t_i} $, $(A \X^{(1)})_i \in F_{t_i} \subseteq F_{\ld-1}$  for all $t_i \leq \ld -1$. Thus, we get $a_{i,m } \a_{\ld} \in F_{\ld -1}$ which is possible if and only if $a_{i,m }=0$ for all $1 \leq i \leq s_{\ld-1} $. Continuing as in the case for $\ld =2$, we get $a_{i,j}=0 $ for all $1 \leq i \leq s_{\ld -1} $ and $j > s_{\ld -1} $. Hence, we can write the matrix $A$ as 
    $$
        A = \begin{bmatrix}
            A'_{11}& \bf{0}\\
            A'_{21}& A'_{22}\\
        \end{bmatrix},
    $$
    where $A'_{11} \in \GL(s_{\ld -1} , F_{\ld -1})$, $A'_{22} \in \GL(\mu_{\ld}, F_{\ld})$ and $A'_{21} \in M_{\mu_{\ld} \times s_{\ld -1} }(F_{\ld})$. 
     Consider 
     $$
        \AA' = F_1 \times \cdots \times F_1 \times \cdots \cdots \times F_{\ld-1} \times \cdots \times F_{\ld-1}  \subseteq \F_q^{s_{\ld-1} }.
    $$
    For $\X \in \AA$, Let $\X' \in \AA'$ obtained from $\X$ by truncating last $\mu_{\ld}$  components. Since $A \X \in \AA$  for all $\X \in \AA$, we obtain $A'_{11} \X'  \in \AA'$ for all $\X' \in \AA'$. Then, by induction hypothesis, $A'_{11}$ has the required form, i.e.,
     $$
        A'_{11} = \begin{bmatrix}
        B_{11} & \bf{0}& \ldots & \bf{0} \\
        B_{21} & B_{22} & \ldots & \bf{0} \\
        \vdots & \vdots & \ddots &\vdots\\
        B_{(\ld-1) 1}& B_{(\ld-1) 2}& \ldots & B_{(\ld-1)(\ld-1) }\\
        \end{bmatrix},
    $$
    where $B_{ij} \in M_{\mu_i \times \mu_j}(F_i) $ and $\det(B_{ii}) \neq 0$. Hence, 
    $$
        A = \begin{bmatrix}
        B_{11} & \bf{0} & \ldots & \bf{0} & \bf{0} \\
        B_{21} & B_{22} & \ldots & \bf{0} & \bf{0} \\
        \vdots & \vdots & \ddots &\vdots& \vdots \\
        B_{(\lambda-1) 1}& B_{(\lambda-1) 2}& \ldots & B_{(\lambda-1)(\lambda-1) }& \bf{0}\\
        A_{\lambda 1} & A_{\lambda_2}& \ldots& A_{\lambda(\lambda-1)} & A_{\lambda \lambda}\\
         \end{bmatrix},
    $$
where $A'_{21} = \begin{bmatrix}
    A_{\lambda 1} & A_{\lambda_2} & \ldots & A_{\lambda(\lambda-1)}\\
\end{bmatrix}$ and $A'_{22} = A_{\lambda \lambda}$. Hence, $\phi \in \H_{\AA}$ and thus, $\Aff(\AA) \subseteq \H_{\AA}$.

\medskip
\noindent
{\bf Case 2:} 
    $\underline{b \neq 0}:$\\
    Let $\phi \in \Aff(\AA)$, i.e., $\phi: \AA \rightarrow \AA$ is a bijection given by $\phi(\X) = A \X+b$ where $A \in \GL(m, \Fq)$, $b \in \F_q^m$. We need to show that $A \in \G$ and $b \in \AA$. Note that $\X_0 = (0, 0, \ldots, 0)^{T} \in \AA$ and by the given hypothesis, $\phi(\X_0) \in \AA$, but $\phi(\X_0) = A X_0 +b = b$. Thus, $b  \in \AA$. Now, let $\Y= A\X+b$, then $\Y \in \AA$ for all $\X \in \AA$. Hence, $A\X =\Y - b \in \AA$ for all $\X \in \AA$. Consider $\psi: \AA \rightarrow \AA$ given by $\psi(\X) =A\X$. Then, $\psi$ is  a bijection and hence by Case 1, $A \in \G$. Thus, we have proved that $A \in \G$, $b \in \AA$. Hence, $\phi \in \H_{\AA}$.
\end{proof}

\begin{cor} \label{|Aff_AA|}
With $\AA$ as defined in equation $\eqref{eqn_AA}$, we have
\begin{equation}
    |\Aff(\AA)| =  \prod_{i=1}^{\lambda} \left( d_{s_i }^{ ^{\mu_i(s_{i-1}  +1)}}  \prod_{t=0}^{\mu_{i}-1} (d_{s_i}^{ ^{\mu_{i}}} - d_{s_i}^{ ^t}) \right).
\end{equation}
\end{cor}

\begin{proof}
We know that $\Aff(\AA) = \{ \sigma_{A,b} \in \AGL(m, \Fq): A \in \G, b \in \AA\}$. Thus, 
\begin{equation*}
    \begin{split}
        \left| \Aff(\AA) \right| &= |\AA|\cdot |\G|\\
        &= |\AA| \cdot \left( \prod_{i=1}^{\lambda} \left|  \GL(\mu_i, F_i) \right| \right) \left(\prod_{t=1}^{\lambda} \left| M_{\mu_t \times s_{t-1}}(F_t) \right| \right) \\ 
        &= \left (\prod_{i=1}^{\lambda} d_{s_i}^{ ^{\mu_i}} \right)  \left( \prod_{i=1}^{\lambda} \left \{ \prod_{t=0}^{\mu_{i}-1} (d_{s_i}^{ ^{\mu_{i}}} - d_{s_i}^{ ^t})\right\} \right) \left(\prod_{i=1}^{\lambda}d_{s_{i}}^{ ^{\mu_{i}   s_{i-1} }} \right)\\
        &= \prod_{i=1}^{\lambda} \left( d_{s_i}^{ ^{\mu_i(s_{i-1} +1)}}  \prod_{t=0}^{\mu_{i}-1} (d_{s_i}^{ ^{\mu_{i}}} - d_{s_i}^{ ^t}) \right)
    \end{split}
\end{equation*}
\end{proof}

\section{Enumeration of minimum weight codewords}\label{section_enumeration}

We have seen from Theorem \ref{thm_carvalho} that minimum weight codewords of $\AC_{q}(u,\AA)$ are given by evaluations of polynomials of $\S_{\leq u}(\AA)$ of the form $f(\X) = \g \cdot h_k^{\Omega}(\sigma(\X))$ for some $\g \in \F_q^{*}, \sigma \in \Aff(\AA)$ and $h_k^{\Omega}(\X)$ is given by 
\begin{equation}\label{eqn_h-k-Omega}
    h_k^{\Omega}(X_1, \ldots,X_m) = \prod_{\underset{i\neq k}{i=1}}^{j+1}(1-X_{i}^{d_i-1})\prod_{\a=1}^{d_{k}-(d_{j+1}-\ell)}(X_k-\w_{\a}),
\end{equation}
for some $1\leq k\leq j+1$ such that $d_k \geq d_{j+1}-\ell $ and $\Omega:= \{ \w_1, \ldots, \w_{d_{k}-(d_{j+1}-\ell)} \} \subseteq A_k$. Using the characterization of the affine group $\Aff(\AA)$ from Proposition \ref{prop_affine  group}, we can write the polynomial $f(\X)$ as follows:
\begin{equation}\label{eqn_f_min}
    f(X_1, \ldots,X_m) = \g  \prod_{\underset{i\neq k}{i=1}}^{j+1}(1-(L_{i}+ b_i)^{d_i-1}) \cdot \prod_{\a=1}^{d_{k}-(d_{j+1}-\ell)}(L_k+b_k-\w_{\a}),
\end{equation}
for some $ \g \in \F_q^{*}$, $1 \leq k \leq j+1$ such that $d_{k} \geq d_{j+1}-\ell$, and linearly independent linear homogeneous polynomials $L_1, \dots , L_{j+1} \in \Fq[X_1, \dots , X_m]$, and scalars $b_i, \omega_i \in \Fq$ such that $ \Omega:= \{ \w_1,\ldots, \w_{d_{k}-(d_{j+1}-\ell)} \} \subseteq A_k$ and 
$$
    L_i \in F_{t_i}[X_1, \dots , X_{s_{t_i}}] \text{ and } b_i \in F_{t_i} \quad \text{whenever } s_{{t_i}-1} < i \leq s_{t_i} \text{ for } 1\leq t_i \leq \lambda.
$$   
Thus, the support of a minimum weight codeword of $\AC_q(u, \AA)$ is given by 
\begin{equation*}\label{eqn_supp_f}
    \supp_{\AA}(\c_f) = \bigg( \bigcap_{\underset{i\neq k}{i=1}}^{j+1} Z_{\AA}\left(L_i + b_{i} \right) \bigg) \bigcap  \bigg( \bigcup_{\w \in A_k \setminus \Omega} Z_{\AA}(L_{k}+b_k-\w) \bigg).  
\end{equation*}
Throughout the article, we will use the notations of equation $\eqref{eqn_f_min}$ to describe polynomials corresponding to minimum weight codewords. Let $\N_{q}(u,\AA)$ denotes the set of minimum weight codewords of $\AC_{q}(u,\AA)$. Then 
$$
    \N_{q}(u, \AA) = \bigcup_{\substack{k=1\\ d_{k} \geq d_{j+1}-\ell}}^{j+1}  \N_{q}^{(k)}(u, \AA),
$$
where
\begin{equation}\label{eqn_N-k}
    \N_{q}^{(k)}(u, \AA) := \bigcup_{\substack{\Omega \subseteq A_k\\ |\Omega| = d_{k}-(d_{j+1}-\ell)}} \N_{q}^{(k, \Omega)}(u, \AA), 
\end{equation}
and for a fixed $k$ and $ \Omega$, $N_{q}^{(k, \Omega)}(u, \AA)$ is given by 
\begin{equation}\label{eqn_N-k-Omega}
    \N_{q}^{(k, \Omega)}(u, \AA)  := \left\{ \c_{f}: f(\X)= \gamma \cdot  h_{k}^{\Omega}(\sigma(\X))  \text{ for some }\g \in \F_q^* \text{ and } \sigma \in \Aff(\AA) \right\}. 
\end{equation}
Therefore, 
\begin{equation}\label{eqn_N}
    \N_{q}(u, \AA) = \bigcup_{\substack{k=1\\ d_{k} \geq d_{j+1}-\ell}}^{j+1}  \left( \bigcup_{\substack{\Omega \subseteq A_k\\|\Omega| = d_{k}-(d_{j+1}-\ell)}}  \N_{q}^{(k, \Omega)}(u, \AA) \right).
\end{equation}
We determine $\left| \N_{q}(u,\AA) \right|$ for $-1 \leq u \leq K$. We know, by convention $\AC_{q}(-1, \AA) = \{ \bf{0}\}$. Therefore, $\N_{q}(-1, \AA) = \emptyset$, and hence 
$$
    \left| \N_{q}(-1, \AA) \right|= 0.
$$
If $u=0$, then $\AC_{q}(u, \AA) = \{ (\lambda, \ldots, \lambda): \lambda\in \Fq \}$. Note that $\d_{q}(0, \AA) =n$. Therefore, 
$\N_{q}(u, \AA) = \left| \{ \lambda (1, \ldots, 1): \lambda \in \F_q^*\} \right|$, and hence 
\begin{equation}\label{eqn_N_0}
    \left| \N_{q}(0, \AA) \right|= q-1.    
\end{equation}
Now, consider $1 \leq u \leq K$. Write 
\begin{equation}\label{eqn_u}
    u = \sum_{i=1}^{j}(d_i-1)+ \ell, \qquad  0 \leq j < m,  \quad 0 < \ell \leq d_{j+1}-1. 
\end{equation}
Define 
\begin{equation}\label{eqn_G_A}
    \GA: = \F_q^* \times  \Aff(\AA) = \{ (\g, \sigma_{A,b}): \g \in \F_q^*, \sigma_{A,b} \in \Aff(\AA) \text{ for some } A \in \G, b \in \AA \},
\end{equation}
where $\G$ is given by definition \ref{def_G}. Then, $\GA$ is a group with respect to component-wise operation. Moreover, 
\begin{equation}\label{eqn_|G_A|}
    |\GA| = |\F_q^*|\cdot |\Aff(\AA)| =(q-1) \prod_{i=1}^{\lambda} \left( d_{s_i}^{ ^{\mu_i(s_{i-1} +1)}}  \prod_{t=0}^{\mu_{i}-1} (d_{s_i}^{ ^{\mu_{i}}} - d_{s_i}^{ ^t}) \right).
\end{equation}
Fix $k, \Omega$ where $1 \leq k \leq j+1$ such that $d_{k} \geq d_{j+1}- \ell$ and  $\Omega \subseteq A_k$ of size $d_{k}-(d_{j+1}-\ell)$. Consider the group action $\Phi$ of the group  $ \GA$ acting on the set  $\N_{q}^{(k, \Omega)}(u, \AA)$ given by 
\begin{equation}\label{eqn_group-action}
    \Phi: \GA \times \N_{q}^{(k, \Omega)}(u, \AA) \rightarrow \N_{q}^{(k, \Omega)}(u, \AA) \quad \text{ given by }  \quad ((\g, \sigma_{A,b}), \c_f) \mapsto \g \cdot \c_g,
\end{equation}
where $g (\X)=f( \sigma_{A,b}(\X)).$ Note that it is a well-defined group action and moreover, it is easy to observe that 
$$
    \N_{q}^{(k, \Omega)}(u,\AA) = \{ \Phi((\g, \sigma_{A,b}), \c_{h_{k}^{\Omega}}): (\g, \sigma_{A,b}) \in \GA\}.
$$
Therefore, we can say 
\begin{equation*}\label{eqn_N=orb}
    \mathcal{N}_{q}^{(k, \Omega)} \left(u,\mathcal{A} \right) = \mathrm{orb}_{\GA}(\c_{h_{k}^{\Omega}}).
\end{equation*}
To determine $\left| \N_{q}^{(k, \Omega)}(u,\AA) \right|$ for all possible choices  for $k$ and $\Omega$, it suffices to determine $\left| \stab_{\GA}(\c_{h_{k}^{\Omega}})\right|$ where 
\begin{equation}\label{eqn_stab_h-k-omega}
    \stab_{\GA}(\c_{h_{k}^{\Omega}}) : = \left\{ (\g, \sigma_{A,b}) \in \GA: \g \cdot h_{k}^{\Omega}(A\X+b) = h_{k}^{\Omega}(\X) \right \}.
\end{equation}

Recall that for $1 \leq i \leq j+1$,  $t_i$ is the unique integer such that $s_{t_{i}-1} < i \leq s_{t_i}$. Define $r := t_{j+1}$.   The following lemmas provide some fundamental conclusions $\stab_{\GA}(\c_{h_{k}^{\Omega}})$. 

\begin{lemma}\label{lemma_stab_1}
For $1 \leq u \leq K$, let $u = \sum\limits_{i=1}^{j}(d_i -1)+ \ell$ for unique $j, \ell$ such that $0\leq j <m$ and $0 <\ell \leq d_{j+1}-1$. Let $k, \Omega, \GA$ be as defined before. If 
    $(\g, \sigma_{A,b}) \in \stab_{\GA}(\c_{h_{k}^{\Omega}})$, then the following holds: 
    \begin{itemize}
        \item[(i)] $a_{i,k} =0$ for $1 \leq i \leq j+1, i \neq k$.
        \item[(ii)] $b_i = 0$ for $1 \leq i \leq j+1, i \neq k$.
        \item[(iii)] $a_{k,k}\neq 0$ and $\g = (a_{k,k})^{-s}$ where $s =d_k -(d_{j+1}-\ell)$. 
    \end{itemize}
\end{lemma}

\begin{proof}
Let  $(\g, \sigma_{A,b}) \in \stab_{\GA}(\c_{h_{k}^{\Omega}})$. Then it must follow that $\g\cdot h_{k}^{\Omega}(A\X+b) = h_{k}^{\Omega}(\X)$, where $A = (a_{i,j}) \in \G$ and  $b =(b_1, \ldots, b_m)^{T} \in \AA$. Since $A \in \G$, we get $a_{i,j} = 0$ whenever $j > s_{t_i}$ for $1 \leq t_i \leq \lambda$. Hence, for $\Omega = \{ \w_1, \ldots, \w_s\}$ where $s =d_{k}-(d_{j+1}-\ell)$, we have   
$$
    \g \prod_{\substack {i=1 \\ i\neq k} }^{j+1}\left(1 - (L_i+b_i)^{d_i-1}\right) \prod_{\a=1}^{s}(L_k + b_k -\omega_{\a}) \pmod{\I(\AA)} = \prod_{\substack {i=1 \\ i\neq k} }^{j+1}\left(1 - X_i^{d_i-1}\right) \prod_{\a=1}^{s}(X_k -\omega_{\a}),
$$
where $L_1, \ldots, L_{j+1}$ are linearly independent linear homogeneous polynomials such that $L_i \in F_{s_t}[X_1, \ldots, X_{s_t}]$ is given by  
$$
    L_{i}(X_1, \ldots, X_{s_{t_i}}) = \sum_{\zeta=1}^{s_{t_i}} a_{i,\zeta } X_{\zeta}, 
$$ 
and $ b_i \in F_{s_{t_i}}$ for $1 \leq i \leq j+1$. Substituting  $X_1= \cdots= X_{k-1}=X_{k+1}= \cdots= X_{s_r} =0 $, we get
$$
  \g  \prod_{\substack {i=1 \\ i\neq k} }^{j+1}\left(1 - (a_{i,k} X_k+b_i)^{d_i-1}\right) \prod_{\a=1}^{s}(a_{k,k} X_k + b_k -\omega_{\a}) \pmod{(X_k^{d_k}-X_k)}=  \prod_{\a=1}^{s}(X_k -\omega_{\a}).
$$
Now, if $\ell < d_{j+1}-1$, then $s \leq d_{k}-2$.  If we assume $a_{k,k} = 0$, then upon comparing the degrees, we get 
$$
    \sum_{\substack{i=s_{t_{k} -1} +1 \\ i \neq k\\ a_{i,k} \neq 0}}^{j+1}\left( d_{i}-1 \pmod{d_{k}}\right) = s,
$$
which is not possible since $\left(d_{i}-1 \pmod{d_{k}}\right) = d_{k}-1$ for all $ s_{t_{k}-1}+1 \leq i  \leq  j+1$. Therefore, we must have $a_{k,k} \neq 0$  which further implies that $a_{i,k} = 0 $ for all $ i \in \{s_{t_{k}-1}+1, \ldots, j+1\}\setminus \{k\}$.  On the other hand, if $\ell = d_{j+1}-1$, then $s =d_{k}-1$ and hence, upon comparing degrees, we get 
$$
   \sum_{\substack{i=s_{t_{k} -1} +1 \\ a_{i,k} \neq 0}}^{j+1}\left( d_{i}-1 \pmod{d_{k}}\right) = d_{k}-1,
$$
which holds true if and only if there exists a unique $i$ such that $s_{t_{k}  -1} < i \leq s_{t_k }$ and $a_{i,k} \neq 0$. Without loss of generality, we assume that $i=k$. Note that $a_{i,k} = 0$ for all $i \in \{1,\ldots, s_{t_{k-1}}\}$ since $A \in \G$. Thus, it follows that $a_{i,k} = 0$ for $ 1 \leq i \leq j+1, i \neq k$. Now, comparing the coefficient of $X_k^s$ on both sides, we get 
$$
    \g \prod_{\substack {i=1 \\ i\neq k} }^{j+1}\left(1 - ( b_i)^{d_i-1}\right) a_{k,k}^s = 1.
$$
Since $b_i \in A_i = \F_{d_i}$, Lemma \ref{lemma_a^s} tells us that $(1 - ( b_i)^{d_i-1} )\neq 0$ if and only if $b_i=0$. Therefore, we must have  $b_i = 0$ for $1 \leq i \leq j+1, i \neq k$ and $\g =  (a_{k,k})^{-s}$.  Hence, the proof is complete. 
\end{proof}

\begin{lemma}\label{lemma_stab_2}
With notations as before, let $(\g, \sigma_{A,b}) \in \stab_{\GA}(\c_{h_{k}^{\Omega}})$. Then  $a_{i,\zeta} =0$ for $s_{r-1} < i \leq j+1$ and $j+2 \leq \zeta\leq s_r.$
\end{lemma}
 
\begin{proof}
Let $(\g, \sigma_{A,b}) \in \stab_{\GA}(\c_{h_{k}^{\Omega}})$. Let $s = d_{k}-(d_{j+1}-\ell)$. Then, it follows that 
\begin{multline*}
    \g \prod_{\substack {i=1 \\ i\neq k} }^{j+1}\left(1 - (\sum_{\zeta =1}^{s_{t_i}} a_{i, \zeta} X_{\zeta}+b_i)^{d_i-1}\right) \prod_{\a=1}^{s} \left(\sum_{\zeta =1}^{s_{t_{k}}} a_{k, \zeta} X_{\zeta} + b_k -\omega_{\a}\right) \pmod{\I(\AA)} = \\ 
    \prod_{\substack {i=1 \\ i\neq k} }^{j+1}\left(1 - X_i^{d_i-1}\right) \prod_{\a=1}^{s}(X_k -\omega_{\a}),
\end{multline*}
where $1 \leq t_i \leq r$ is the unique integer such that $s_{t_{i}-1}< i \leq s_{t_i}$ and $\Omega = \{\omega_1, \ldots, \omega_{s}\}\subseteq A_k$.  Using Lemma \ref{lemma_stab_1}, the above equation reduces to 
\begin{multline*}
    \g \prod_{\substack {i=1 \\ i\neq k} }^{j+1}\left(1 - \left(\sum_{\substack{\zeta =1\\ \zeta \neq k}}^{s_{t_i}} a_{i, \zeta} X_{\zeta}\right)^{d_i-1}\right) \prod_{\a=1}^{s} \left(\sum_{\substack{\zeta =1\\\zeta \neq k}}^{s_{t_{k}}} a_{k, \zeta} X_{\zeta} + a_{k,k} X_k + b_k -\omega_{\a}\right) \pmod{\I(\AA)} = \\ 
    \prod_{\substack {i=1 \\ i\neq k} }^{j+1}\left(1 - X_i^{d_i-1}\right) \prod_{\a=1}^{s}(X_k -\omega_{\a}),
\end{multline*}
where $a_{k,k} \neq 0$. 
Since $s =d_k - (d_{j+1}-\ell)$ and $0< \ell \leq d_{j+1}-1$, it follows that $0 \leq s \leq d_{k}-1$ and hence $A_{k}\setminus\{\Omega\} \neq \emptyset$. Let $\widehat{\omega} \in A_{k}\setminus\{\Omega\}$. We first consider the case when $s> 0$. If $a_{\rho, \eta} \neq 0 $ for some $\rho \in \{s_{r-1}+1, \ldots, j+1\}$ and $\eta \in \{j+2, \ldots, s_{r}\}$, then evaluating the above equation at $P = (p_1, \ldots, p_m) \in \AA$ given by 
$$
p_{\zeta} = \begin{cases}
    1 &{\text{ if } \zeta =k},\\
    \frac{1}{a_{\rho, \eta}}\left(\omega_1 -b_k -a_{k,k} \widehat{\omega}\right) &{\text{ if } \zeta=\eta},\\
    0 &{\text{ otherwise, } }\\
\end{cases}
$$
we obtain that $LHS = 0$ whereas $RHS \neq 0$, which is a contradiction. On the other hand, if $s =0$, then the equation becomes 
$$
    \g \prod_{\substack {i=1 \\ i\neq k} }^{j+1}\left(1 - \left(\sum_{\substack{\zeta =1\\ \zeta \neq k}}^{s_{t_i}} a_{i, \zeta} X_{\zeta}\right)^{d_i-1}\right) \pmod{\I(\AA)} =  \prod_{\substack {i=1 \\ i\neq k} }^{j+1}\left(1 - X_i^{d_i-1}\right). 
$$
If $a_{\rho, \eta} \neq 0 $ for some $\rho \in \{s_{r-1}+1, \ldots, j+1\}$ and $\eta \in \{j+2, \ldots, s_{r}\}$, then evaluating the above equation at $e_{\eta} = (0, \ldots, 0, 1, 0, \ldots, 0)$, where $1$ occurs at $\eta^{\textrm{th}}$ position, then $LHS=0$, whereas $RHS=1$, which again is a contradiction. Hence, we must have 
$a_{i, \zeta} = 0$ for all $i \in \{ s_{r-1}+1, \ldots, j+1 \} $ and $\zeta \in \{j+2, \ldots, s_r\}$. 
\end{proof}

\begin{lemma}\label{lemma_stab_3}
With notations as before, let $(\g, \sigma_{A,b}) \in \stab_{\GA}(\c_{h_{k}^{\Omega}})$. Consider the  matrix $A' = (a_{ij}^{'}) \in \G$ where
$$
    a'_{i,j} = \begin{cases}
        a_{i,j} & \text{ if } i \neq k,\\
        a_{k,j} & \text{ if } i = k \text{ and } j \in \{s_{t_k}+1, \ldots, m \} \cup \{k\}.\\
    \end{cases}
$$ 
Then $(\g, \sigma_{A',b}) \in \stab_{\GA}(\c_{h_{k}^{\Omega}})$. 
\end{lemma}

\begin{proof}
We first note that $A'$ differs from $A$ only in some entries of $k^{th}$-row. Since $A'\in \G$, we have $(\g, \sigma_{A',b}) \in \GA$. Let $g(\X) \in \Fq[X_1, \ldots, X_m]$ be given by $g(\X) = \g \cdot h_{k}^{\Omega} (A' \X+b)$. Then $\c_g \in \N^{(k,\Omega)}(u, \AA)$. To show $(\g, \sigma_{A',b}) \in \stab_{\GA}(\c_{h_k^{\Omega}})$, it suffices to show that $f(P) = g(P) = h_{k}^{\Omega}(P)$ for all $P \in \AA$ where $f(\X) = \g \cdot h_{k}^{\Omega}(A \X+b)$. Since $(\g, \sigma_{A,b}) \in \stab_{\GA}(\c_{h_k^{\Omega}})$, Lemma \ref{lemma_stab_1} gives $a_{i,k} = 0 = b_i$ for all $i \in [j+1]\setminus \{k\}$,  $a_{k,k} \neq 0$ and $\g = (a_{k,k})^{-s}$ where $s = d_k -(d_{j+1}-\ell)$. Therefore, We can write 
\begin{equation*}
    \begin{split}
        g(\X) &=  \g  \prod_{\substack{i=1\\ i \neq k}}^{j+1} \left( 1-\left(\sum_{\zeta=1}^{s_{t_i}} a'_{i, \zeta} X_{\zeta}\right)^{d_i-1} \right) \prod_{\a=1}^{s}\left( \sum_{\zeta =1}^{s_{t_k}} a'_{k ,\zeta } X_\zeta + b_k -\omega_{\a}\right) \\
        &=  (\g a_{k,k}^s ) \prod_{\substack {i=1 \\ i\neq k} }^{j+1}\left(1 - \left( \sum_{\substack {\zeta=1 \\ \zeta \neq k} }^{s_{t_i}} a_{i, \zeta } X_\zeta \right)^{d_i-1}\right) \prod_{\a=1}^{s}\left( \left(\sum_{\substack{\zeta =1\\\zeta \neq k}}^{s_{t_k}} \frac{a'_{k, \zeta }}{a_{kk}} X_\zeta \right) +X_k -\omega_{\a}\right),\\
    \end{split}
\end{equation*}
where $\g a_{k,k}^s = 1$. We know that $(\g, \sigma_{A,b}) \in \stab_{\GA}(\c_{h_{k}^{\Omega}})$ which implies $\c_{f} = \c_{h_{k}^{\Omega}}$ and hence $f(P) = h_{k}^{\Omega}(P)$ for all $P \in \AA$. Let $S = 
 \supp(\c_{h_{k}^{\Omega}})$. Then  
$$
    S = \{ P = (p_1, \ldots, p_m) \in \AA: p_i =0 \text{ for  all } i \in [j+1]\setminus\{k\}, p_k \in A_k \setminus \Omega\}.
$$
Now, if $P \in S$, then  
$$
    g(P) =  \prod_{\a=1}^{s}(p_k -\omega_{\a}) = h_{k}^{\Omega}(P) = f(P).
$$ 
If $P \in \AA \setminus S$, we must have $h_{k}^{\Omega}(P) = f(P) = 0$ since 
\begin{multline*}
    \AA \setminus S = \{(0, \ldots,0,  p_k, 0,\ldots, 0, p_{j+2}, \ldots, p_m): p_k \in \Omega \} \\
    \bigcup \{(p_1, \ldots, p_m): (p_1, \ldots,p_{k-1},p_{k+1}, \ldots, p_{j+1}) \neq (0,\ldots, 0)   \}. 
\end{multline*}
Therefore, it suffices to show that $g(P) = 0$  for all $P  \in \AA \setminus S$. Let there exists $P \in \AA \setminus S$ such that $g(P) \neq 0$. But, it holds true if and only if $\sum\limits _{\substack{\zeta=1 \\ \zeta \neq k}}^{s_{t_i}} a_{i, \zeta} p_{\zeta} = 0 $ for  $1 \leq i \leq j+1, i \neq k $ and $\sum\limits_{\substack{\zeta =1 \\ \zeta \neq k}}^{s_{t_k}} \frac{a'_{k, \zeta }}{a_{k,k}} p_\zeta \neq p_k -\omega_{\a}$ for all $1 \leq \a \leq  s$. It follows from Lemma \ref{lemma_stab_2} that 
$a_{i,\zeta} =0$ for $s_{r-1} < i \leq j+1$ and $j+2 \leq \zeta\leq s_r$. Therefore, the system of equations  $\sum\limits _{\substack{\zeta=1 \\ \zeta \neq k}}^{s_t} a_{i, \zeta} X_{\zeta} = 0 $ is a homogeneous system of $j$ linearly independent equations in $j$ variables, namely $X_1, \ldots, X_{k-1}, X_{k+2}, \ldots, X_{j+1}$. Hence, it admits a trivial solution, i.e., $p_{\zeta} = 0$ for  $ 1 \leq \zeta \leq j+1, \z \neq k$. Now, substituting these values in $g(P)$, we get 
$$
    g(P) = \prod_{\a=1}^{s}(p_k -\omega_{\a})=0,
$$
since $P \in \AA\setminus S$ and $p_i = 0$ for all $i \in [j+1]\setminus\{k\}$, which is a contradiction. Hence, $g(P) = 0 = h_{k}^{\Omega}(P)$ for all $P \in \AA$. Therefore, $(\g, \sigma_{A',b}) \in \stab_{\GA}(\c_{h_k^{\Omega}})$.
\end{proof}

We now determine $\left| \N_{q}^{(k, \Omega)}(u,\AA) \right|$ for all possible choices of $k$ and $\Omega$ and consider different cases accordingly. 

\subsection{Case $\ell = d_{j+1}-1$}\label{subsection_case_1}
For $1 \leq u \leq K$, we write 
$ u = \sum\limits_{i=1}^{j}(d_i -1) + \ell$, for uniquely determined integers $j, \ell$ such that $0 \leq j < m$ and $0 < \ell \leq d_{j+1}-1$. In this section, we consider the case when $\ell = d_{j+1}-1$. Then, we can write  
$$
    u = \sum_{i=1}^{j+1}(d_i-1).
$$
Hence, for $\Omega = \{\w_1, \ldots, \w_{d_{k}-1} \}$, we have 
$$
    h_{k}^{\Omega}(\X) = \prod_{\substack{i=1\\i \neq k}}^{j+1} (1 - X_i^{d_i-1}) \prod_{\a =1}^{d_{k}-1}(X_k -\w_{\a}).
$$
Let $\Omega' = \{\w'_1, \ldots, \w'_{d_k -1}\} \subseteq A_k$ such that $\Omega \neq \Omega'$. Then there exist $a_{k,k}, b_k \in A_k$ such that $a_{k,k} \neq 0$ and $\Omega' = a_{k,k} \Omega+ b_k = \{a_{k,k} \w_i +b_{k}: 1 \leq i \leq d_{k}-1 \}$. Indeed, define 
$$
     a_{k,k} = 1, \quad b_k = \w'_0 - \w_0, 
$$
where $\{ \w_0\} = A_k \setminus \Omega$ and  $\{ \w'_0\} = A_k \setminus \Omega' $. Then, it follows that for $\g=1, A =I_m$ and $b= (0, \ldots, 0, b_k, 0 \ldots, 0)$, we get  $\g \cdot h_{k}^{\Omega}(A\X+b)= h_{k}^{\Omega'}(\X)$ and hence
$$
    \N_{q}^{(k, \Omega)}(u, \AA) = \N_{q}^{(k, \Omega')}(u, \AA).
$$
Therefore, without loss of generality, let $\Omega = A_k \setminus\{0\}$ and since $A_k$ is a subfield of $\Fq$ of size $d_k$, it follows that 
$$ 
    \prod_{\a =1}^{d_k-1}(X_k -\w_{\a}) = 1 - X_k^{d_k -1}.
$$
Hence, we can write 
$$
    h_{k}^{\Omega}(\X) = \prod_{\substack{i=1\\i \neq k}}^{j+1} (1 - X_i^{d_i-1}) \prod_{\a =1}^{d_{k}-1}(X_k -\w_{\a}) = \prod_{i=1}^{j+1} (1 - X_i^{d_i-1}). 
$$
Note that $h_{k}^{\Omega}(\X)$ is independent of $k$ and $\Omega$. Thus, every minimum weight codeword $\c_f$ of $\AC_{q}(u, \AA)$ is given by $\c_f^{} = \g \cdot \c_{h_{j+1} \circ \sigma}^{}$ for some $\g \in \F_q^*$ and $\sigma \in \Aff(\AA)$ and $h_{j+1}(\X) \in \S_{\leq u}(\AA)$ is given by

\begin{equation}\label{eqn_h}
    h_{j+1}(\X) = \prod_{i=1}^{j+1} (1 - X_i^{d_i-1}).
\end{equation}
Thus, 
\begin{equation}\label{eqn_N(l=d-1)}
     \N_{q}(u, \AA) = \orb_{\GA}(\c_{h_{j+1}}).
\end{equation}

\begin{defn}
    Let $t$ be a positive integer such that $1 \leq t \leq m$. Define 
    $$\M_{q}(t,m): = \{ A =(a_{i,j}) \in \GL(m , \Fq): a_{i,j} = 0 \text{ for } 1 \leq i \leq t, t+1 \leq j \leq m \}.$$
\end{defn}

\noindent
Note that $\M_{q}(t,m)$ is a collection of block lower triangular matrices with non-singular diagonal blocks of size $t \times t$ and $(m-t) \times (m-t)$  respectively. Moreover, $M_{q}(t,m)$ is a subgroup of $\GL(m, \Fq)$ and 

\begin{equation}\label{eqn_M_q_t}
    |\M_{q}(t,m)| = q^{t(m-t)} \left(\prod_{i=0}^{t-1}(q^t -q^i) \right) \left( \prod_{\z=0}^{m-t-1}(q^{m-t} -q^{\z})\right).
\end{equation}
Before stating the main result, we prove a lemma that will be used later in this section. 

\begin{lemma}\label{lemma_BLT}
For a positive integer $t$ such that $ 1 \leq t \leq m$, let $h_{t} \in \S$ be given by 
$$
    h_t(X_1, \ldots, X_m) = \prod_{i=1}^{t}(1-X_i^{q -1}).
$$
Consider the set $\mathcal{V}_{q}(t,m) := \left\{(\g, \sigma_{A,b})\in \F_q^* \times \AGL(m, \Fq) : \g \cdot \c_{ h_{t}\circ \sigma_{A,b}} = \c_{h_{t}} \right\}$. Then,
$$
    (\g, \sigma_{A,b}) \in \mathcal{V}_{q}(t,m) \iff    \g = 1, A \in \M_{q}(t,m), \text{ and } b_1 = \cdots =b_t =0.
$$
\end{lemma}

\begin{proof}
The proof directly follows from Lemma \ref{lemma_stab_1} and Lemma \ref{lemma_stab_2}. 
\end{proof}
  
\begin{cor}\label{cor_V-t}
With notations as above, 
$$
    |\mathcal{V}_{q}(t,m)| = q^{(t+1)(m-t)} \left(\prod_{i=0}^{t-1}(q^t -q^i) \right)\left(\prod_{\z=0}^{m-t-1}(q^{m-t} -q^{\z}) \right).
$$
\end{cor}

\begin{proof}
It follows from Lemma \ref{lemma_BLT}, that $(\g, \sigma_{A,b}) \in \mathcal{V}_{q}(t,m)$ if $\g =1, A \in \M_{q}(t,m)$, and $b_i = 0 $ for $1 \leq i \leq t$. Now, it is easy to observe that there is a unique choice for $\g$, $q^{m-t}$ choices for $b$. Since $A \in \M_{q}(t,m)$, $A$ can have as many choices as  $|\M_{q}(t)|$. Therefore, 
$$
    |\mathcal{V}_{q}(t,m)| = q^{m-t} |\M_{q}(t)| =q^{m-t} \left[q^{t(m-t)} \left(\prod_{i=0}^{t-1}(q^t -q^i) \right)\left(\prod_{i=0}^{m-t-1}(q^{m-t} -q^i) \right)\right].
$$
\end{proof}


\begin{theorem}\label{thm_l=d-1}
Let $\AA$ be as defined in equation \eqref{eqn_AA} and let $1\leq u \leq K = \sum\limits_{i=1}^{m} (d_i-1)$ be such that $u =\sum\limits_{i=1}^{j+1}(d_i -1)$ for some  $0\leq j < m$. 
Let $r =t_{j+1}$. Then 
$$
    |\N_{q}(u, \AA)| = (q-1) \left( \prod_{i=1}^{j+1} d_i \right) \stirling{\mu_r}{j+1 - s_{r-1}}_{d_{j+1}}.
$$
\end{theorem}

\begin{proof}
It follows from equation \eqref{eqn_N(l=d-1)} that 
\begin{equation*}
    \N_{q}(u, \AA) = \orb_{\GA} (\c_{h_{j+1}}),
\end{equation*}
where $h_{j+1}(\X)$ is given by equation $\eqref{eqn_h}$.  Using Theorem \ref{thm_orb-stab}, it suffices to determine $|\stab_{\GA}(\c_{h_{j+1}})|$ where 
$$
    \stab_{\GA}(\c_{h_{j+1}}) = \{(\g, \sigma_{A,b}) \in \GA : \g \cdot h_{j+1}(A \X+b) = h_{j+1}(\X)\}.
$$
We prove the result by induction on $r$. If $r=1$, then $j+1 \leq s_1 =\mu_1$ and hence, for $\c_f \in \N_{q}(u, \AA)$, $f$ is given by 
$$
    f(\X) =  \g \prod_{i=1}^{j+1}\left( 1-(L_i +b_i)^{d_i-1}\right),
$$
where $L_1, \ldots, L_{j+1} \in F_1[X_1,\ldots, X_{s_1 }] $ are linearly independent linear homogeneous polynomials and $b_i \in F_1$ for $1 \leq i \leq j+1$. Then, it follows from  Lemma \ref{lemma_BLT} that $(\g, \sigma_{A,b}) \in \stab_{\GA}(\c_{h_{j+1}})$ if and only if $\g = 1, A_{11} \in \M_{d_{s_1}}(j+1, m)$ and $b_1 = \cdots= b_{j+1} =0$. Therefore,  using equation $\eqref{eqn_M_q_t}$ and equation \eqref{eqn_|G_A|}, we get 
$$
    \left| \stab_{\GA}(\c_{h_{j+1}}) \right| = \frac{d_{s_1}^{m-(j+1)}  \left | \M_{d_{s_1 }}(j+1, \mu_1)  \right | \left|\GA\right|} { \left| \GL(\mu_1, d_{s_1}) \right|}, 
$$
and hence
$$
    \left| \N_{q}(u, \AA) \right| = \frac{|\GA|}{\left| \stab_{\GA}(\c_{h_{j+1}}) \right|} = \frac{(q-1) \left( d_{s_1 }^{j+1} \right)|\GL(\mu_1, d_{s_1 })|} {|\M_{d_{s_1 }}(j+1,\mu_1)|}.
$$
Since $d_1 = \cdots = d_{s_1 }$, we get 
\begin{equation*}
        \left| \N_{q}(u, \AA) \right| 
        = (q-1) \left(\prod_{i=1}^{j+1} d_i \right) \frac{|\GL(\mu_1, d_{s_1 })|}{ |\M_{d_{s_1 }}(j+1,\mu_1)|} = (q-1) \left(\prod_{i=1}^{j+1} d_i \right) \stirling{\mu_1}{j+1}_{d_{j+1}}.
\end{equation*}
Thus, the result is proved for $r=1$. Now, let us assume that the result holds for $r-1$. Since $s_{r-1} < j+1 \leq s_r$, we can write  
$$
    f = \g \prod_{i=1}^{s_1} \left( 1-(L_i +b_i)^{d_i-1}\right)  \prod_{i=s_1+1}^{j+1}\left( 1-(L_{i} +b_{i})^{d_{i}-1}\right),
$$
where $L_i  = \sum\limits_{\zeta=1}^{s_{t_i}} a_{i,\zeta} X_{\zeta} \in F_{t_i}[X_1, \ldots, X_{s_{t_i}}]$ and $b_i \in F_{t_i}$. Note that for $1 \leq i \leq s_1$, $L_i +b_i =0$ is a system of $s_1$ linearly independent equations in $s_1$ variables over $F_1$. Hence, it has a unique solution for each choice of the vector $(\v_1, \ldots, \v_{s_1})^{T}$.  Now, since $X_1, \ldots, X_{s_1}$ have fixed values, we can rewrite the system of linear equations as $L'_{i} =-b'_i$ for $s_1 +1 \leq i \leq j+1$ where $L'_i = \sum\limits_{\zeta=s_1 +1}^{s_{t_i}} a'_{i, \zeta} X_\zeta$ for $2 \leq t_i \leq \lambda$. Define
$$
    \AA' = F_2 \times\cdots \times F_2 \times \cdots \times F_{\lambda} \times \cdots \times F_{\lambda}.
$$
Then, we can write
$$
    f = \g  \left(\prod\limits_{i=1}^{s_1} \left( 1-(L_i +b_i)^{d_i-1}\right) \cdot g\right),
$$ 
where 
$$
    g =  \prod\limits_{i=s_1+1}^{j+1}\left( 1-(L_{i} +b_{i})^{d_{i}-1}\right) \in \S_{(u-\mu_1(d_1 -1))}(\AA).
$$
Hence, for fixed $\nu_1, \ldots, \nu_{s_1} \in F_1 = \F_{d_{s_1}}$, we have 
$$
    g = \prod\limits_{i=s_1+1}^{j+1}\left( 1-(L'_{i} +b'_{i})^{d_{i}-1}\right) \in \S_{(u-\mu_1(d_1 -1))}(\AA').
$$
which corresponds to a minimum weight codeword of $\AC_{q}(u-\mu_1(d_1 -1), \AA')$. Therefore, by induction hypothesis, 
$$
    |\N_{q}(u-\mu_1(d_1-1), \AA')| = (q-1)\left( \prod_{i=s_1 +1}^{j+1} d_i\right) \stirling{\mu_r}{j+1-s_{r-1}}_{d_{j+1}}.
$$
But, there are $\prod\limits_{i=1}^{s_1}d_i $ choices for the vector $(\nu_1, \ldots, \nu_{s_1})^{T}$. Therefore, 
\begin{equation}\label{eqn_result_l=d-1}
    |\N_{q}(u, \AA)| = (q-1)\left( \prod_{i=1}^{j+1} d_i\right) \stirling{\mu_r}{j+1-s_{r-1}}_{d_{j+1}}.
\end{equation}
Hence, the proof is complete.
\end{proof}

\subsection{Case $\ell < d_{j+1}-1$} \label{subsection_case_2}

In this section, we consider the case when $1 \leq u \leq K$ is given by  
\begin{equation}\label{eqn_u_l <d-1}
    u = \sum\limits_{i=1}^{j}(d_i -1)+ \ell \text{ for some } 0 \leq j < m \text{ and } 0 < \ell < d_{j+1} -1.
\end{equation}
It follows from Theorem \ref{thm_carvalho} and Proposition \ref{prop_affine  group} that every minimum weight codeword is given by $\c_f$ for some $f \in \S_{\leq u}(\AA)$ given by equation $\eqref{eqn_f_min}$. Let $1 \leq k_{0} \leq j+1$ be the least integer $k$ such that $d_{k} \geq d_{j+1}-\ell$. For each $k_0 \leq k \leq j+1$, let $\N_{q}^{(k)}(u,\AA)$ be as defined in equation \eqref{eqn_N-k}. Then, we can write 
\begin{equation}\label{eqn_union_N_k}
    \N_{q}(u, \AA) = \bigcup_{k=k_0}^{j+1} \N_{q}^{(k)}(u, \AA).
\end{equation}
Recall, for $1 \leq k \leq m$, let $t_k$ be the unique integer such that $s_{t_{k} -1} < k \leq s_{t_{k}}$ and let $r=t_{j+1}$. 

\begin{prop}\label{prop_empty-intersection}
Let $u$ be given by equation \eqref{eqn_u_l <d-1}. With notations as before, let $r =t_{j+1}$ and let $1 \leq k_{0} \leq j+1$ be the least integer $k$ such that $d_{k} \geq d_{j+1}-\ell$. Let $ k_0 \leq k_1 <  k_2 \leq j+1$. Then, 
$$
    \N_{q}^{(k_1)}(u, \AA) = \N_{q}^{(k_2)}(u, \AA) \text{ if and only if } t_{k_1}  = t_{k_2} .
$$
Furthermore, when $t_{k_1}  < t_{k_2} $,
$$
    \N_{q}^{(k_1)}(u, \AA) \bigcap \N_{q}^{(k_2)}(u, \AA) =\emptyset.
$$
\end{prop}

\begin{proof}
For the sake of notation, let $\N^{(k_i)}:= \N_{q}^{(k_i)}(u, \AA)$ for $i=1,2$. Let $t_{k_1}  = t_{k_2} $. It suffices to show that $\N^{(k_1)} \subseteq \N^{(k_2)}$. Let $\c_f \in \N^{(k_1)}$. Then $f =\g_1\cdot  (h_{k_1}^{\Omega_1} \circ \sigma_1)$ for some $\g_1 \in \F_q^{*}$ and $\sigma_1 \in \Aff(\AA)$ and some $\Omega_1 \subseteq A_{k_1}$ such that $ |\Omega_1| = d_{k_1}-(d_{j+1}-\ell)$. Consider a permutation $\tau \in S_{n}$ such that $\tau$ acts on $\F_q^m$ as follows:
$$
\tau (X_1,\ldots,X_m) = (X_1, \ldots,X_{k_{1} -1}, X_{k_2},X_{k_{1} +1}, \ldots \ldots, X_{k_{2}-1}, X_{k_1}, X_{k_{2} +1}\ldots,X_m).
$$
Note that $\tau$ interchanges the variables $X_{k_1}$ and $X_{k_2}$. The permutation $\tau$ can be expressed in the matrix form by $\tau (X) = P_{k_1, k_2} X$ where $P_{k_1, k_2}$ is an $m \times m$ permutation matrix. Since $t_{k_1}  = t_{k_2} $, $P_{k_1, k_2}$ is a non singular block lower triangular matrix and thus, $P_{k_1, k_2} \in \G$. Hence, $\tau \in \Aff(\AA)$. Therefore, 
$$
  \c_{f \circ \tau} = \g_1 \cdot \c_{h_{k_1}^{\Omega_1} \circ \sigma_1 \circ \tau} = \g_1 \cdot \c_{h_{k_2}^{\Omega_1} \circ \sigma_1}   \in \N^{(k_2)}.
$$
The above equality is independent of the choice of  $\Omega_1 \subseteq A_{k_1}$. Hence, $\c_f \in \N^{(k_2)}$ and it follows that $\N^{(k_1)} = \N^{(k_2)}$ whenever $t_{k_1}  = t_{k_2} $.

On the other hand, let $t_{k_1}  < t_{k_2} $. We need to show that $\N^{(k_1)} \cap \N^{(k_2)} = \emptyset$. Assuming on the contrary, let $\c_f \in \N^{(k_1)} \cap \N^{(k_2)}$ for some $f (\neq 0) \in S_{\leq u}(\AA)$. Then $\c_f = \g_1 \cdot \c_{g_1}$ where $g_1(\X) = h_{k_1}^{\Omega_1}(\sigma_1(\X))$  and $\c_f  = \g_2 \cdot \c_{g_2}$ where $g_2(\X) = h_{k_2}^{\Omega_2}(\sigma_2(\X))$ for some $\g_1, \g_2 \in \F_q^{*}$,  $\sigma_1, \sigma_2 \in \Aff(\AA)$ and $\Omega_i \subseteq A_{k_i}$ of size $d_{k_i}-(d_{j+1}-\ell)$ for $i=1,2$. Then, $\c_f = \g_1 \cdot \c_{g_1} = \g_2 \cdot \c_{g_2}$ and hence, it must follow that 
$$ 
    g_1(\X) = \g \cdot g_2(\X) \pmod{\I(\AA)},
$$
for some $\g \in \F_q^*$. Then, we can write
\begin{equation}\label{eqn_compare_hk1_hk2}
    h_{k_1}^{\Omega_1}(\X) = \g \cdot  h_{k_2}^{\Omega_2}(\sigma(\X))  \pmod{\I(\AA)} 
\end{equation}
for some  $\g \in \F_q^*$ and $\sigma \in \Aff(\AA)$. Using Proposition \ref{prop_affine  group}, we get 
\begin{equation}\label{eqn_h_k_2}
    \g \cdot  h_{k_2}^{\Omega_2}(\sigma(\X))  = \prod_{\substack{i = 1 \\ i\neq k_2} }^{j+1}  \left( 1 -(L_i +b_i)^{d_i -1} \right) \prod_{\a=1}^{d_{k_2}-(d_{j+1}-\ell)} (L_{k_2}+b_{k_2} -\omega_{\a}), 
\end{equation}
for some linearly independent linear polynomials $L_1, \ldots, L_{j+1}$ such that $L_i \in F_{t_i}[X_1, \ldots, X_{s_{t_i}}]$ and $b_i \in F_{t_i}$  for $1 \leq t_i \leq r$ and  $ \Omega = \{ \w_1,\ldots, \w_{d_{k_2}-(d_{j+1}-\ell)} \} \subseteq A_{k_2}$.

We compare monomials of degree $u$ in equation $\eqref{eqn_compare_hk1_hk2}$. Let $\a_{i,\d}$ denotes the exponent of $X_{\d}$ in the function $ \left( 1 -(L_i +b_i)^{d_i -1} \right) \pmod{\I(\AA)}$ for $i \in \{1, \ldots, j+1 \} \setminus \{k_2\}$ and $\a_{k_2,\d}$ denotes the exponent of $X_{\d}$ in the function $\prod\limits_{s\a=1}^{d_{k_2}-(d_{j+1}-\ell)} (L_{k_2} -\omega_{\a})\pmod{\I(\AA)}$. Note that, $\a_{i, \d} =0 $ for all $\d> s_{{t_i}}.$ Thus, monomials of degree $u$ in $\g \cdot h_{k_2}^{\Omega_2}(\sigma(\X)) \pmod{\I(\AA)}$ are given by 
$$
    \left\{ \prod_{\underset{i\neq k_2}{i=1}}^{j+1} \left(  \prod_{\z =1}^{s_t} X_{\z}^{\a_{i,\z}}\right) \right\} \prod_{\z=1}^{s_{t_{k_2} } } X_{k_2}^{\a_{k_2,\z}},
$$
where $1 \leq t_i \leq r $ such that  
\begin{equation}\label{eqn_sum1}    
    \sum_{\z=1}^{s_t} \a_{i,\z} = \begin{cases}
        d_i -1 & \text{ if } i \in \{1, \ldots, j+1 \} \setminus \{ k_2\}, \\
        d_{k_2}-(d_{j+1}-\ell) & \text{ if }i = k_2.\\
    \end{cases}
\end{equation}
Note that there is only one monomial of degree $u$ in $h_{k_1}^{\Omega_1}(\X)$ and is given by 
$$
    \left(\prod\limits_{\substack{\zeta = 1\\ \zeta \neq k_1}}^{j+1} X_{\zeta}^{d_{\zeta-1}} \right)X_{k_1}^{d_{k_1}-(d_{j+1}-\ell)}.
$$
Comparing with monomials of $\g \cdot h_{k_2}^{\Omega_2}(\sigma(\X)) \pmod{\I(\AA)}$, we get 
\begin{equation}\label{eqn_sum2}
    \sum_{i=1}^{j+1} \a_{i,\z} = 
    \begin{cases}
        d_{\z} -1 & \text{ if } \z \in \{1, \ldots, j+1 \} \setminus \{ k_1\}, \\
        d_{k_1}-(d_{j+1}-\ell) & \text{ if } \z = k_1,\\
        0 & \text{ if } \z \in \{j+2, \ldots, s_r \}.
    \end{cases}
\end{equation}
Writing the exponents in matrix form, we get $C' := (\a_{i,\z})$, which is a $j+1 \times s_r$ matrix. Extend $C'$ to $C$ by adding $s_r -(j+1)$ rows of zeroes at the bottom such that $C$ is an $s_r \times s_r$ square matrix. It follows that $C = [C_{\rho, \eta }]_{1\leq \rho, \eta \leq r}$ is a block lower triangular matrix with block $C_{\rho, \eta } $ of size $\mu_{\rho} \times \mu_{\eta}$. Define by $S_{\rho,\eta}$, the sum of entries of block $C_{\rho, \eta}$, i.e.,
$$
    S_{\rho, \eta} = \sum_{\z=s_{\eta -1} +1 }^{s_{\eta}} \ \sum_{i=s_{\rho-1}+1 }^{s_{\rho}} \a_{i, \z}.
$$
Note that $S= (S_{\rho, \eta})$ is a $r \times r$ lower triangular matrix, i.e., $S_{\rho, \eta } = 0$ for $\eta > \rho$. Comparing equations \eqref{eqn_sum1} and \eqref{eqn_sum2}, we get
\begin{eqnarray*}
    S_{1,1} &= \mu_1 (d_{s_1}-1).\\
    S_{1,1}+S_{2,1}+\cdots+S_{r,1} &= \mu_1 (d_{s_1}-1). 
\end{eqnarray*}
But $S_{\rho, \eta} \geq 0$ for all $1 \leq \rho, \eta \leq r$. Therefore, $S_{2,1} = \cdots = S_{r,1} = 0$. Moreover, we have   
\begin{eqnarray*}
    S_{2,1} +S_{2,2} &= \mu_2 (d_{s_2}-1).\\
    S_{2,2}+S_{3,2}+\cdots+S_{r,2} &= \mu_2 (d_{s_2}-1). 
\end{eqnarray*}
But $S_{2,1} = 0$ implies $S_{\rho ,2} = 0$ for $\rho > 2$. 
Continuing this way,  for $\rho < t_{k_1}$, we get 
$$
    S_{\rho, \eta } = \begin{cases}
    \mu_{\rho} (d_{s_{\rho}}-1) & \text{ if } \rho = \eta,\\
    0 & \text{ if } \rho \neq \eta.\\   
\end{cases}
$$
Again, we observe that 
$$
    S_{t_{k_1},t_{k_1}}  = \mu_{t_{k_1}} (d_{s_{t_{k_1}}}-1) \quad \text{and} \quad  \sum_{\rho = t_{k_1}+1}^{r} S_{\rho, t_{k_1}} = (\mu_{t_{k_1}} -1) (d_{s_{t_{k_1}}}-1) + \left(d_{k_1}- (d_{j+1}-\ell)\right).
$$
But $\ell < d_{j+1}-1$ implies $d_{k_1}- (d_{j+1}-\ell) < (d_{s_{t_{k_1}}}-1)$. Hence, it follows that
$$
    \sum_{\rho = t_{k_1}+1}^{r} S_{\rho, t_{k_1}} < 0, 
$$
which is a contradiction since $S_{\rho ,\eta} \geq 0$ for all $1 \leq \rho ,\eta \leq r$. Hence, $\N^{(k_1)} \cap \N^{(k_2)} =\emptyset$.
\end{proof}

Applying Proposition \ref{prop_empty-intersection}, we can rewrite equation \eqref{eqn_union_N_k} as

$$  
    \N_{q}(u, \AA) = \N_{q}^{({j+1})}(u,\AA) \bigcup \left(\bigcup_{t=t_{k_0}}^{r-1} \N_{q}^{(s_{t} )}(u,\AA)\right),
$$
and moreover, it follows that
\begin{equation}\label{eqn_N_sum}
    | \N_{q}(u, \AA)| =|\N_{q}^{(j+1)}(u,\AA)| +  \sum_{t=t_{k_0}}^{r-1} |\N_{q}^{(s_{t} )}(u,\AA)|.
\end{equation}
The next step is to determine $|\N_{q}^{(k)}(u,\AA)|$ for all $1 \leq k \leq j+1$ where 
$$
   \N_{q}^{(k)}(u,\AA) = \bigcup_{\substack{\Omega \subseteq A_k \\ |\Omega| =d_{k}-(d_{j+1}-\ell)}} \N_{q}^{(k, \Omega)}(u, \AA)  
$$
Therefore, we determine $\left| \N_{q}^{(k, \Omega)}(u,\AA) \right|$ for all possible choices  for $k$ and $\Omega$. 
We further consider two sub-cases when $t_k = r$ and $t_k < r$ respectively. 

\subsubsection{Subcase $t_k =r$}

Let $1 \leq u \leq K$ be such that $u =\sum\limits_{i=1}^{j}(d_i -1) + \ell$, $0 \leq j < m$ and $ 0 < \ell < d_{j+1} -1$. If $t_k =r$, $d_k =d_{j+1}$. Then, it follows from Proposition $\ref{prop_empty-intersection}$ that
$$
    \N_{q}^{(k)}(u, \AA) = \N_{q}^{(j+1)}(u, \AA).
$$
Therefore, every minimum weight codeword is given by the evaluation $\c_f$ for the polynomial $f \in S_{\leq u}(\AA)$ given by 
$$
    f(\X) = \g \cdot h_{j+1}^{\Omega}(A\X+b) = \g \prod_{i=1}^{j}\left(1 - \left(\sum_{\z=1}^{s_{t_i}} a_{i, \z} X_{\z} +b_i \right)^{d_i-1} \right) \prod_{\a =1}^{\ell} \left(\sum_{\zeta =1}^{s_r} a_{j+1, \zeta} X_{\zeta} +b_{j+1} - \w_{\a}\right), 
$$
for some $\g \in \F_q^*, A =(a_{i,j}) \in \G, b \in \AA$ and  $\Omega = \{\w_1, \ldots, \w_{\ell} \} \subseteq A_{j+1}$.
The next theorem gives $|\N_{q}^{(k)}(u,\AA)|$ when $t_k =r$.

\begin{theorem}\label{thm_tk=r}
Let $1 \leq u \leq K$ be such that $u =\sum\limits_{i=1}^{j}(d_i -1) + \ell$, $0 \leq j < m$ and $ 0 < \ell < d_{j+1} -1$. Let $r, k,t_k, k_0$ be as before. If $t_{k}=r$, then  
$$
    |\N_{q}^{(k)} (u, \AA)| = (q-1) \bigg(\prod\limits_{\underset{i\neq k}{i=1}}^{j+1} d_i \bigg) \displaystyle{\stirling{\mu_r}{j-s_{r-1}}_{d_{j+1}}  \stirling{s_{r}-j}{1}_{d_k}} \binom{d_{k}}{d_{j+1}-\ell}.
$$
\end{theorem}

\begin{proof}
We have $t_k =r$, and thus $d_k =d_{j+1}$. Then, it follows from Proposition $\ref{prop_empty-intersection}$ that
$$
    \N_{q}^{(k)}(u, \AA) = \N_{q}^{(j+1)}(u, \AA).
$$
Moreover, equation $\eqref{eqn_N-k}$ implies
$$
    \N_{q}^{(j+1)}(u, \AA) = \bigcup_{\substack{\Omega \subseteq A_{j+1}\\ |\Omega| = \ell}}  N_{q}^{(j+1, \Omega)}(u, \AA).
$$
With respect to group action $\Phi$ defined by equation $\eqref{eqn_group-action}$, we know that
$$
    N_{q}^{(j+1, \Omega)}(u, \AA) = \orb_{\GA}(\c_{h_{j+1}^{\Omega}}),
$$
for a fixed subset $\Omega = \{\w_1, \ldots, \w_{\ell} \}$ of $A_{j+1}$ of size $\ell$ and $h_{j+1}^{\Omega}$ is given by 
$$
  h_{j+1}^{\Omega}(\X) = \prod_{i=1}^{j}(1-X_i^{d_i-1}) \prod_{\a=1}^{\ell}(X_{j+1}-\w_{\a}).  
$$
Therefore, using Theorem \ref{thm_orb-stab}, it suffices to determine $|\stab_{\GA}(\c_{h_{j+1}^{\Omega})}|$. We prove the result by induction on $r$. If $r=1$, then $j+1 \leq s_1 =\mu_1$, and hence for $\c_f \in \N_{q}^{(j+1, \Omega)}(u, \AA)$, $f$ is given by
$$
    f(X) = \g \prod_{i=1}^{j} \left( 1-(L_i +b_i)^{d_i-1}\right) \prod_{\a=1}^{\ell}(L_{j+1}+b_{j+1}-\w_{\a}),
$$
where $L_1, \ldots, L_{j+1}$ are linearly independent linear homogeneous polynomials such that $L_i \in F_1[X_1, \ldots, X_{s_1}]$, $b_i \in F_1$ and $\Omega = \{\w_1, \ldots,\w_{\ell}\} \subseteq A_{j+1}$. Let $(\g, \sigma_{A,b}) \in \stab_{\GA}(\c_{h_{j+1}^{\Omega}})$, then 
$$
 \g  \prod_{i=1}^{j} \left( 1-(L_i +b_i)^{d_i-1}\right) \prod_{\a=1}^{\ell}(L_{j+1}+b_{j+1}-\w_{\a}) \pmod{\I(\AA)} = \prod_{i=1}^{j}(1-X_i^{d_i-1}) \prod_{\a=1}^{\ell}(X_{j+1}-\w_{\a}).
$$
By Lemma \ref{lemma_stab_1}, it follows that $a_{i,j+1} =0 =b_i$ for $1 \leq i \leq j$, $a_{j+1,j+1} \neq 0$ and $\g = (a_{j+1,j+1})^{- \ell}$. 
Moreover, Lemma \ref{lemma_stab_2} gives $a_{i, \z} =0$ for $1\leq i \leq j+1$ and $j+2 \leq \z \leq s_1$. Therefore, we get 
\begin{multline*}
    \g  \prod_{i=1}^{j}\left(1 - \left( \sum_{\substack{\zeta=1\\ \zeta \neq j+1}}^{j} a_{i, \zeta} X_{\zeta} \right)^{d_i-1}\right) \prod_{\a=1}^{\ell}\left( \sum_{\z=1}^{j+1} a_{j+1, \z } X_{\z} + b_{j+1} -\omega_{\a}\right) \pmod{\I(\AA)} =\\
    \prod_{i=1}^{j}\left(1 - X_i^{d_i-1}\right) \prod_{\a=1}^{\ell}(X_{j+1} -\omega_{\a}).
\end{multline*}
Upon substituting $X_1 = \cdots = X_{j} =X_{j+2} =\cdots = X_{s_1} =0$, we get 
$$ 
    \g  (a_{j+1, j+1})^{\ell} \prod_{\a=1}^{\ell}\left[ X_{j+1}- \left( \frac{\omega_{\a}-b_{j+1}}{a_{j+1, j+1}}\right) \right] =  \prod_{\a=1}^{\ell}(X_{j+1} -\omega_{\a}),
$$
which holds true if and only if $\g  (a_{j+1 ,j+1})^{\ell} =1$ and $\left\{ \frac{\omega_{1}-b_{j+1}}{a_{j+1, j+1}}, \ldots, \frac{\omega_{\ell}-b_{j+1}}{a_{j+1, j+1}} \right\} = \{\omega_1, \ldots, \omega_{\ell} \}$. This is equivalent to determining $\sigma_{a_{j+1, j+1}, b_{j+1}} \in \AGL(1, F_{1})$ such that $a_{j+1, j+1} \Omega + b_{j+1} =\Omega$. Define the set 
$$
    \Delta_{\Omega} := \left\{\sigma_{a_{j+1, j+1},b_{j+1}} \in\AGL(1, F_{1}) : \{a_{j+1, j+1}\w_1 + b_{j+1}, \ldots, a_{j+1, j+1 }\w_{\ell} + b_{j+1}\} = \{\w_1, \ldots, \w_{\ell} \} \right\}.
$$
Therefore, our equation becomes 
\begin{multline*}
    \prod_{i=1}^{j}\left(1 - \left( \sum_{\substack {\zeta=1 \\ \zeta \neq j+1} }^{j} a_{i, \zeta} X_{\zeta} \right)^{d_i-1}\right) \prod_{\a=1}^{\ell}\left( \left(\sum_{\zeta =1}^{j} \frac{a_{j+1,  \zeta }}{a_{j+1, j+1}} X_\zeta \right) +X_{j+1} -\omega_{\a}\right) \pmod{\I(\AA)} = \\
    \prod_{i=1}^{j}\left(1 - X_i^{d_i-1}\right) \prod_{\a=1}^{\ell}(X_{j+1} -\omega_{\a}).
\end{multline*}

Define a matrix $A' = (a_{i,\z}^{'}) \in M_{m \times m}(\Fq)$ given by 
$$
    a'_{i, \z} = \begin{cases}
        a_{i, \z}, & \text{ if } i \neq j+1,\\
        a_{j+1, \z} & \text{ if } i = j+1 \text{ and } \z \in \{j+1, \ldots, m \}. \\
    \end{cases}
$$
Note that $A'$ differs from $A$ only in first $j$ entries of $(j+1)^{th}$-row.  Since $(\g, \sigma_{A,b}) \in \stab_{\GA}(\c_{h_{j+1}^{\Omega}})$, Lemma \ref{lemma_stab_3} gives $(\g, \sigma_{A',b}) \in \stab_{\GA}(\c_{h_{j+1}^{\Omega}})$. Define 
$$
    \widehat{\AA_{j+1}} = A_1 \times \cdots \times A_j \times A_{j+2} \times \cdots \times A_m \subseteq \F_q^{m-1}.
$$
Also, define $\widehat{\X_{j+1}} = (X_1, \ldots, X_j, X_{j+2}, \ldots, X_m) \in \F_q^{m-1}$. Then, we can write
\begin{multline*}
    \stab_{\GA}(\c_{h_{j+1}^{\Omega}}) = \{ (\g, \sigma_{A,b}) \in \GA: a_{i, j+1}=0 = b_i \text{ for } 1 \leq i \leq j,  \\
    a_{j+1,\z} =0 \text{ for } j+2 \leq \z \leq s_1, \sigma_{a_{j+1, j+1}, b_{j+1}} \in \Delta_{\Omega} \text{ and } \widehat{h_{j+1}} (\sigma_{\widehat{A_{j+1}},{\bf 0}}(\widehat{\X_{j+1}})) = \widehat{h_{j+1}}(\widehat{\X_{j+1}}) \}.
\end{multline*}
where $\widehat{A_{j+1}}$ is obtained from $A$ by removing the $(j+1)^{\rm th}$-row and $(j+1)^{\rm th}$-column of $A$ and $\sigma_{\widehat{A_{j+1}},{\bf 0}} \in \Aff(\widehat{\AA_{j+1}})$ and $\widehat{h_{j+1}} \in \Fq[\widehat{\X_{j+1}}]$ is given by 
\begin{eqnarray*}
    \widehat{h_{j+1}}(\widehat{\X_{j+1}}) = \prod_{i=1} ^{j}\left(1 - X_i^{d_i-1}\right). 
\end{eqnarray*}

Note that $\widehat{h_{j+1}}$ is a polynomial of degree $(u-\ell)$ and corresponds to a minimum weight codeword of $\AC_{q}(u-\ell, \widehat{\AA_{j+1}})$. For $(\g, \sigma_{A,b}) \in \stab_{\GA}(\c_{h_{j+1}^{\Omega}})$, it should necessarily hold that $\widehat{h_{j+1}} (\sigma_{\widehat{A_{j+1}},{\bf 0}}(\widehat{\X_{j+1}})) = \widehat{h_{j+1}}(\widehat{\X_{j+1}})$. But Lemma \ref{lemma_BLT} tells us that this is true if and only if $(\widehat{A_{j+1}})_{11} \in \M_{d_{s_1}}(j, \mu_1-1)$, where $(\widehat{A_{j+1}})_{11}$ is the first diagonal block of $\widehat{A_{j+1}}$.  Therefore, 
\begin{multline*}
    \stab_{\GA}(\c_{h_{j+1}^{\Omega}}) = \{ (\g, \sigma_{A,b}) \in \GA: \g (a_{j+1, j+1})^{\ell} =1,  \sigma_{a_{j+1, j+1}, b_{j+1}} \in \Delta_{\Omega}, a_{i, j+1}=0 = b_i \\ \text{ for } 1 \leq i \leq j,
    a_{j+1, \z} =0    \text{ for } j+2 \leq \z \leq s_1,  (\widehat{A_{j+1}})_{11} \in \M_{d_{s_1}}(j, \mu_1-1) \}, 
\end{multline*}
and hence, it follows from equation \eqref{eqn_|G_A|} and \eqref{eqn_M_q_t} that 
$$
     |\stab_{\GA}(\c_{h_{j+1}^{\Omega}})| =  (d_{j+2} \cdots d_m) \frac{ d_{j+1}^{j} (d_{j+2} \cdots d_{s_1})|\Delta_{\Omega}| |\M_{d_{s_1}}(j, \mu_1-1)| |\GA| } {|\GL(\mu_1, d_{s_1})|}. 
$$
Thus, we have 
\begin{equation*}
    \begin{split}
       \left| \N_{q}^{(j+1, \Omega)}(u, \AA) \right|
       &= \frac{|\GA|}{|\stab_{\GA}(\c_{h_{j+1}^{\Omega}})|}\\
       &= \frac{(q-1) (d_{1}\cdots d_{m}) |\GL(\mu_1, d_{s_1})| }{d_{j+1}^{j} (d_{j+2}\cdots d_m) (d_{j+2} \cdots d_{s_1})  |\Delta_{\Omega}| |\M_{d_{s_1}}(j, \mu_1-1)| }\\
       &= \frac{(q-1) (d_{1}\cdots d_{m}) (d_{j+1}^{\mu_1}-1) d_{j+1}^{\mu_1 -1} |\GL(\mu_1-1, d_{j+1})| }{d_{j+1}^{j} (d_{j+2}\cdots d_m) (d_{j+2} \cdots d_{s_1})  |\Delta_{\Omega}| |\M_{d_{s_1}}(j, \mu_1-1)| }\\
       &= \frac{(q-1) d_{j+1} (d_{j+1}^{\mu_1}-1) }{|\Delta_{\Omega}|} \left(\prod\limits_{i=1}^{j} d_i\right) \stirling{\mu_1-1}{j}_{d_{j+1}} \\
        &= \frac{(q-1) d_{j+1} (d_{j+1}^{\mu_1-j}-1) }{  |\Delta_{\Omega}|} \left(\prod\limits_{i=1}^{j} d_i\right) \stirling{\mu_1}{j}_{d_{j+1}}\\
       &= \frac{\b_{j+1}}{|\Delta_{\Omega}|},
    \end{split}
\end{equation*}
where 
$$
    \b_{j+1} = (q-1) d_{j+1}  (d_{j+1}^{\mu_1-j}-1) \left(\prod\limits_{i =1}^{j} d_i\right) \stirling{\mu_1}{j}_{d_{j+1}}.
$$
Now, observe that $\b_{j+1}$ is independent of choice of $\Omega$ and hence $|\N_{q}^{(j+1, \Omega)}|$ depends only on the size of the set $\Delta_{\Omega}$. Since there are $\eta$ distinct orbits $\Omega_1, \ldots, \Omega_{\eta}$, we must have
$$
    \left| \N_{q}^{(j+1)}(u, \AA)\right| = \sum_{i=1}^{\eta} \left| \N_{q}^{(j+1, \Omega_i)}(u, \AA)\right| =\sum_{i=1}^{\eta} \frac{\b_{j+1}}{|\Delta_{\Omega_i}|}.  
$$
Using equation \eqref{eqn_stab-reciprocal-sum}, we get 
$$
    \left| \N_{q}^{(j+1)}(u, \AA)\right| = \frac{\b_{j+1}}{d_{j+1}(d_{j+1} -1)} \binom{d_{j+1}}{\ell}.
$$
Therefore,
$$
    \left| \N_{q}^{(j+1)}(u, \AA)\right|  = (q-1) \bigg(\prod\limits_{i=1}^{j} d_i \bigg) \displaystyle{\stirling{\mu_1}{j}_{d_{j+1}}  \stirling{\mu_1-j}{1}_{d_{j+1}}}\binom{d_{j+1}}{d_{j+1}-\ell}. 
$$
Hence, the result if proved for $r=1$. Let us assume that the result is true for $r-1$. Since $s_{r-1} < j+1 \leq s_r$, we can write  
$$
    f = \g \prod_{i=1}^{s_1} \left( 1-(L_i +b_i)^{d_i-1}\right)  \prod_{i=s_1+1}^{j}\left( 1-(L_{i} +b_{i})^{d_{i}-1}\right) \prod_{\a=1}^{\ell}(L_{j+1}+b_{j+1}-\w_{\a}),
$$
where $L_i  = \sum\limits_{\zeta=1}^{s_{t_i}} a_{i\zeta} X_{\zeta} \in F_{t_i}[X_1, \ldots, X_{s_{t_i}}]$ and $b_i \in F_{t_i}$ for  $1 \leq t_i \leq r$. Note that for $1 \leq i \leq s_1$, $L_i +b_i =0$ is a system of $s_1$ linearly independent equations in $s_1$ variables over $F_1$. Hence, it has a unique solution for each choice of the vector $(\v_1, \ldots, \v_{s_1})^{T}$.  Now, since $X_1, \ldots, X_{s_1}$ have fixed values, we can rewrite this system of linear equations to be $L'_{i} =-b'_i$ for $s_1 +1 \leq i \leq j$ where $L'_i = \sum\limits_{\zeta=s_1 +1}^{s_{t_i}} a'_{i, \zeta} X_\zeta$  for $2 \leq t_i \leq \lambda$ and similarly replace $(L_{j+1}+b_{j+1}-\w_{\a})$ by $(L'_{j+1}+b'_{j+1}-\w_{\a})$ where $L'_{j+1} = \sum\limits_{\zeta=s_1+1}^{r} a_{j+1, \zeta} X_{\zeta}$. Let
$$
    \AA' = F_2 \times\cdots \times F_2 \times \cdots \times F_{\lambda} \times \cdots \times F_{\lambda}.
$$
Then, 
$$
    f = \g  \left(\prod\limits_{i=1}^{s_1} \left( 1-(L_i +b_i)^{d_i-1}\right) \cdot g\right),
$$ 
where 
$$
    g = \prod\limits_{i=s_1+1}^{j+1}\left( 1-(L_{i} +b_{i})^{d_{i}-1}\right) \left( \prod_{\a=1}^{\ell}(L_{j+1}+b_{j+1}-\w_{\a}) \right)\in \S_{(u-\mu_1(d_1 -1))}(\AA).
$$
Hence, for fixed $\nu_1, \ldots, \nu_{s_1}$, we have 
$$
    g = \prod\limits_{i=s_1+1}^{j+1}\left( 1-(L'_{i} +b'_{i})^{d_{i}-1}\right)\prod_{\a=1}^{\ell}(L'_{j+1}+b'_{j+1}-\w_{\a})  \in \S_{(u-\mu_1(d_1 -1))}(\AA').
$$
By induction hypothesis, 
$$
    |\N_{q}^{(j+1)}(u-\mu_1(d_1-1), \AA')| = (q-1)\left( \prod_{i=s_1 +1}^{j} d_i\right) \stirling{\mu_r}{j-s_{r-1}}_{d_{j+1}} \stirling{s_r -j}{ 1}_{d_{j+1}} \binom{d_k}{d_{j+1}-\ell}.
$$
But, there are $\left(\prod\limits_{i=1}^{s_1}d_i\right)$ choices for the vector $(\nu_1, \ldots, \nu_{s_1})^{T}$. Hence, 
\begin{equation}\label{eqn_result_tk=r}
    |\N_{q}^{(j+1)}(u, \AA)| = (q-1)\left( \prod_{i=1}^{j} d_i\right) \stirling{\mu_r}{j-s_{r-1}}_{d_{j+1}} \stirling{s_r -j}{ 1}_{d_{j+1}} \binom{d_k}{d_{j+1}-\ell}.
\end{equation}
\end{proof}

\subsubsection{Subcase $t_k <r$}
We are now left with the case when  $\ell < d_{j+1} -1$ and $t_{k} < r$. To consider this case, we further divide it into two sub-cases, i.e., $d_k=d_{j+1}-\ell$ and $d_k> d_{j+1}-\ell$  respectively and consider them separately. We first consider the case when  $d_{k} = d_{j+1}-\ell$. Then, we have 
$$
    u = \sum\limits_{\substack{i=1 \\i \neq k}}^{j+1} (d_i-1),
$$
and thus every minimum weight codeword is given by $\c_f = \g \cdot \c_{\widehat{h_{k}} \circ \sigma}$ for some $\g \in \F_q^*, \sigma \in \Aff(\AA) $ where $\widehat{h_k}(\X)$ is given by 
$$ 
    \widehat{h_{k}}(\X) = \sum\limits_{\substack{i=1 \\i \neq k}}^{j+1} \left( 1 - X_{i}^{d_i-1} \right).
$$
The next theorem determines $|\N_{q}^{(k)}(u, \AA)|$ in this case.  

\begin{theorem} \label{thm_d_k=d_j+1 -l}
Let $u = \sum\limits_{i=1}^{j}(d_i-1)+\ell $ for some $0 \leq j <m$, $0 < \ell < d_{j+1}-1$. Let $r, t_k, k_0$ be defined as before, let $1\leq k \leq j+1 $ be such that $t_{k}< r$ and $d_{k} = d_{j+1} -\ell$, then  
$$
    \N_{q}^{(k)}(u, \AA) = (q-1) \bigg(\prod\limits_{\underset{i \neq k}{i=1}}^{j+1} d_i \bigg) \bigg(\prod\limits_{i=s_{t_k}+1}^{j+1} d_i \bigg) \displaystyle{\stirling{\mu_r}{j+1-s_{r-1}}_{d_{j+1}}  \stirling{\mu_{t_k}}{1}_{d_k}}.
$$
\end{theorem}

\begin{proof}
Recall group action  $\Phi$ defined in equation \eqref{eqn_group-action}, i.e., 
$$
    \Phi:  \GA  \times \N_{q}^{(k, \Omega)(u, \AA)} \rightarrow \N_{q}^{(k, \Omega)}(u, \AA) \text{ given by } \Phi((\g, \sigma_{A,b}), \c_f) = \g \cdot \c_{f \circ \sigma_{A,b}}.
$$
We have already seen that 
$$
    \N_{q}^{(k, \Omega)}(u, \AA) = \orb_{\GA}(\c_{h_{k}^{\Omega}})   
$$
for $h_{k}^{\Omega}(\X)$ defined by equation \eqref{eqn_h-k-Omega}. But $d_{k} = d_{j+1}-\ell$ gives $h_{k}^{\Omega}(\X) = \widehat{h_k}(\X)$ and hence is independent of the choice of $\Omega$. Thus, it follows from equation \eqref{eqn_N-k} that $\N_{q}^{(k)}(u, \AA) =  \N_{q}^{(k, \Omega)}(u, \AA)$. Then, by Theorem \ref{thm_orb-stab}, it suffices to determine  $|\stab_{\GA}(\c_{\widehat{h_k}})|$ where 
$$
    \stab_{\GA}(\c_{\widehat{h_k}}) := \{(\g, \sigma_{A,b}) \in \GA:  \g \cdot \widehat{h_{k}}(\sigma_{A,b}(\X)) \pmod{\I(\AA)}=\widehat{h_{k}}(\X)\}.
$$
Let $(\g, \sigma_{A,b}) \in \stab_{\GA}(\c_{\widehat{h_k}})$.  It follows from Lemma \ref{lemma_stab_1} that $a_{i,k} =0 =b_i$ for $1 \leq i \leq j+1, i\neq k$ and $\g = a_{k,k}^{0} = 1$.  Moreover, \begin{equation}\label{GA_to_G}
    \left| \stab_{\GA}(\c_{h_k}) \right| = d_k \left(\prod_{i=j+2}^{m}d_i\right) \left|\stab_{\G}(\c_{\widehat{h_k}}) \right|,
\end{equation}
where $\stab_{\G}(\c_{\widehat{h_k}}): = \{A \in \G: \Phi((1, \sigma_{A,{\bf 0}} ), \c_{\widehat{h_k}}) = \c_{\widehat{h_k}} \} = \{ A \in \G: \c_{\widehat{h_k}}(A\X) = \c_{\widehat{h_k}}(\X)\}$.

Now, consider the matrix $\widehat{A_{k}}$ obtained from $A$ be removing the $k^{\rm th}$-row and $k^{\rm th}$-column of $A$.  Define 
$$
    \widehat{\AA_{k}}: = F_1^{\mu'_1} \times \cdots \times F_1^{\mu'_{\lambda}} \subseteq \F_{q}^{m-1},
$$ 
such that for $1 \leq i \leq \lambda$,
$$
    \mu'_i = \begin{cases}
    \mu_i & {\text{ if } i \neq t_k},\\
    \mu_{t_{k} }-1 & {\text{ if } i = t_k}.
    \end{cases}
$$
Define $s'_0 = 0$ and $s'_t : = \mu'_1 + \cdots + \mu'_t$  for $1 \leq t \leq \lambda$. Define 
$$
    \widehat{\G_{k}}:= \{ \widehat{A_{k}}: A\in \G\} ;
$$  
$$
    \Aff( \widehat{\AA_{k}}):= \left \{\sigma_{\widehat{A_{k}},\widehat{b_{k}}} \in \AGL(m-1, \Fq):  \widehat{A_{k}} \in  \widehat{\G_{k}}, \hspace{1mm}  \widehat{b_{k}} \in  \widehat{\AA_{k}}\right \} ;
$$
and 
$$
    \GG_{\widehat{\AA_{k}}}:= \F_q^{*} \times \Aff(\widehat{\AA_{k}}).
$$
Note that  $\deg(\widehat{h_k}(\X)) = \sum\limits_{\substack{i=1\\ i\neq k}}^{j+1}(d_i-1) =u$ and hence $\c_{\widehat{h_k}} \in \AC_{q}(u,  \widehat{\AA_{k}})$. Define a map 
$$
    \xi: \stab_{\G}(\c_{\widehat{h_k}}) \rightarrow \stab_{\widehat{\G_{k}}}(\c_{\widehat{h_k}}) \text{  given by } \xi(A) =  \widehat{A_{k}}.
$$
Note that $\xi$ is a surjective homomorphism. Moreover, it follows from part $(i)$ of Lemma \ref{lemma_stab_1} that
$$
    \left|\ker(\xi) \right| = d_{k}^{s_{t_k}-1} (d_k -1) \left( \prod\limits_{i=j+2}^{m}d_i \right).
$$ 
Hence, 
\begin{equation}\label{eqn_stab_G}
    \left| \stab_{\G}(\c_{\widehat{h_k}}) \right| = d_{k}^{s_{t_k}-1} (d_k -1) \left( \prod\limits_{i=j+2}^{m}d_i \right) \left| \stab_{\widehat{\G_{k}}}(\c_{\widehat{h_k}}) \right|.
\end{equation}
Note that, $\c_{\widehat{h_{k}}} \in \AC_{q}(u, \widehat{\AA_k})$ and therefore equation $\eqref{GA_to_G}$ implies 
\begin{equation}\label{eqn_stab_GA}
    \left| \stab_{\GG_{ \widehat{\AA_{k}}}}(\c_{\widehat{h_k}}) \right| = \left( \prod\limits_{i=j+2}^{m}d_i \right) \left| \stab_{\widehat{\G_{k}}}(\c_{\widehat{h_k}}) \right|
\end{equation}
But $u = \sum\limits_{\substack{i=1 \\i \neq k}}^{j}(d_i-1)+ (d_{j+1}-1)$ implies $\c_{\widehat{h_k}}$ is a minimum weight codeword of $\AC_{q}(u, \widehat{\AA_{k}})$ corresponding to the case $\ell = d_{j+1}-1$. Thus, 
$$
    \left| \N_{q}(u,  \widehat{\AA_{k}}) \right| = (q-1)\left( \prod_{\substack{i=1 \\i \neq k}}^{j+1} d_i\right) \stirling{\mu_r}{j+1-s_{r-1}}.
$$
Hence, using the orbit-stabilizer theorem for $\c_{\widehat{h_k}} \in \AC_{q}(u, \widehat{\AA_k})$, we get
\begin{equation}\label{eqn_stab_G_A_k}
    \left| \stab_{\GG_{\widehat{\AA_{k}}}}(\c_{h_k}) \right| = \frac{\left| \GG_{\widehat{\AA_{k}}}\right|}{\left| \N_{q}(u, \widehat{\AA_{k}}) \right|},
\end{equation}
where 
$$
    \left | \GG_{\widehat{\AA_{k}}} \right| = (q-1) \prod_{i=1}^{\lambda} \left[d_{s'_i}^{ ^{\mu'_i(s'_{i-1}+1)}} \left\{ \prod_{t=0}^{\mu'_{i}-1} d_{s_i}^{\mu'_i} - d_{s_i}^{t} \right\} \right].     
$$
Thus, equation \eqref{GA_to_G} gives
$$
    \left|\N^{(k)}(u, \AA) \right| = \frac{\left|\GA\right|}{\left|\stab_{\GA}(\c_{\widehat{h_k}})\right|} = \frac{\left|\GA\right|}{d_k \left(\prod\limits_{i=j+2}^{m}d_i\right) \left|\stab_{\G}(\c_{\widehat{h_k}}) \right|}.
$$
But equation \eqref{eqn_stab_G} gives $ \left| \stab_{\G}(\c_{\widehat{h_k}}) \right| =  \left(\prod\limits_{i=j+2}^{m}d_i\right) \left| \stab_{\widehat{\G_{k}}}(\c_{\widehat{h_k}}) \right|$. Therefore, it follows from equations \eqref{eqn_stab_GA} and \eqref{eqn_stab_G_A_k} that 
\begin{equation*}
    \begin{split}
        \left|\N^{(k)}(u, \AA) \right|
        &= \frac{\left|\GA\right| \left( \prod\limits_{i=j+2}^{m}d_i \right)}{d_{k}^{s_{t_k}} (d_k -1) \left( \prod\limits_{i=j+2}^{m}d_i \right)^{2}\left| \stab_{\GG_{\widehat{\AA_{k}}}}(\c_{\widehat{h_k}}) \right|}\\
        &= \frac{\left|\GA\right| \left|\N_{q}(u,\widehat{\AA_{k}})\right|}{d_{k}^{s_{t_k}} (d_k -1) \left( \prod\limits_{i=j+2}^{m}d_i \right) \left|\GG_{\widehat{\AA_k}}\right| } \\
        &= (q-1) \bigg(\prod\limits_{\underset{i \neq k}{i=1}}^{j+1} d_i \bigg) \bigg(\prod\limits_{i=s_{t_k}+1}^{j+1} d_i \bigg) \displaystyle{\stirling{\mu_r}{j+1-s_{r-1}}_{d_{j+1}}  \stirling{\mu_{t_k}}{1}_{d_k}}.\\
    \end{split}
\end{equation*}
Hence, the desired result is obtained. 
\end{proof}
      
Now, we consider the case when $d_{k} > d_{j+1}-\ell$.  Define $s = d_k -(d_{j+1}-\ell)$. Then $1 \leq s \leq d_k -2$. In this case, a minimum weight codeword is represented by evaluation of the polynomial $h_k^{\Omega}(\X)$ (up to $\AA$-equivalence) given by
$$
    h_k^{\Omega}(\X) = \prod_{i=1}^{j+1}(1-X_i^{d_i-1}) \prod_{\a =1}^{s}(X_k-\omega_{\a}),
$$
where $1 \leq k \leq j+1$ such that $d_k > d_{j+1}-\ell$ and $\Omega = \{\w_1, \ldots, \w_{s} \} \subseteq A_k$ of cardinality $s$. 
    
\begin{theorem}\label{thm_d_k>d_j+1 -l}
Let $u = \sum\limits_{i=1}^{j}(d_i-1)+\ell $ for some $0 \leq j <m$, $0 < \ell < d_{j+1}-1$. With $r, t_k, k_0$ be defined as before. Let $1\leq k \leq j+1 $ be such that $t_{k}< r$ and $d_{k} >  d_{j+1} -\ell$, then
$$
    \left| \N_{q}^{(k)}(u, \AA)\right|  = (q-1) \bigg(\prod\limits_{\underset{i \neq k}{i=1}}^{j+1} d_i \bigg) \bigg(\prod\limits_{i=s_{t_k}+1}^{j+1} d_i \bigg) \displaystyle{\stirling{\mu_r}{j+1-s_{r-1}}_{d_{j+1}}  \stirling{\mu_{t_k}}{1}_{d_k}}\binom{d_k}{d_{j+1}-\ell}.
$$
\end{theorem}      
      
\begin{proof}
Let $\mathcal{P}_{s}(A_k)$ denotes the collection of all subsets of $A_k$ having exactly $s$ elements. Then, it follows from equation \eqref{eqn_N-k} and \eqref{eqn_N-k-Omega} that 
$$
    \N_{q}^{(k)}(u,\AA) = \bigcup_{\Omega \in \mathcal{P}_{s}(A_k)}\N_{q}^{(k, \Omega)}(u,\AA),
$$ 
where 
$$
    \{ \g \cdot \c_{h_{k}^{\Omega} \circ \sigma_{A,b}}: \g \in \F_{q}^{*}, \sigma_{A,b} \in \Aff(\AA) \}, 
$$
for some fixed set  $\Omega \in \mathcal{P}_{s}(A_k)$. We first determine $|\N_{q}^{(k,\Omega)}(u,\AA)|$ for a fixed set $\Omega = \{\omega_1, \ldots, \omega_s\}$. Consider the group action $\Phi$ defined by equation \eqref{eqn_group-action}. Then it is easy to observe that $N_{q}^{(k,\Omega)}(u,\AA) = \orb_{\GA}(\c_{h_{k}^{\Omega}})$. Hence, using the orbit-stabilizer theorem, it suffices to determine $|\stab_{\GA}(\c_{h_{k}^{\Omega}})|$ where
$$
    \stab_{\GA}(\c_{h_{k}^{\Omega}}) = \{(\g, \sigma_{A,b}) \in \GA: \g \cdot h_{k}^{\Omega}(\sigma_{A,b}(\X)) \pmod{\I(\AA)}=h_{k}^{\Omega}(\X)\}.
$$
Let $(\g, \sigma_{A,b}) \in \stab_{\GA}(\c_{h_{k}^{\Omega}})$. Then, the following equation must hold, i.e, 
$$
    \g \prod_{\substack {i=1 \\ i\neq k} }^{j+1}\left(1 - (L_i+b_i)^{d_i-1}\right) \prod_{\a=1}^{s}(L_k + b_k -\omega_{\a}) \pmod{\I(\AA)} = \prod_{\substack {i=1 \\ i\neq k} }^{j+1}\left(1 - X_i^{d_i-1}\right) \prod_{\a=1}^{s}(X_k -\omega_{\a}),
$$
where $L_1, \ldots, L_{j+1}$ are linearly independent linear homogeneous polynomials such that $L_i \in F_{s_{t_i}}[X_1, \ldots, X_{s_{t_i}}]$ is given by  
$$
    L_{i}(X_1, \ldots, X_{s_{t_i}}) = \sum_{\zeta=1}^{s_{t_i}} a_{i, \zeta } X_{\zeta}, 
$$ 
and $ b_i \in F_{s_{t_i}}$ for $1 \leq i \leq j+1$. It follows from Lemma \ref{lemma_stab_1} and Lemma \ref{lemma_stab_2} that $a_{i,\z} =0 =b_i $ for $1 \leq i \leq j+1, i \neq k$, $\gamma = (a_{k,k})^{-s}$ and $ a_{i,\z} =0 $ for $s_{r-1}<i \leq j+1$ , $j+2 \leq \z \leq s_r$. Therefore, the equation becomes
\begin{multline*}
     \prod_{\substack {i=1 \\ i\neq k} }^{j+1}\left(1 - \left( \sum_{\substack {\zeta=1 \\ \zeta \neq k} }^{s_{t_i}} a_{i, \z } X_\zeta \right)^{d_i-1}\right) \prod_{\a=1}^{s}\left( X_k +  \sum_{\substack{\zeta =1\\
      \z \neq k}}^{s_{t_k}} \frac{a_{k, \zeta }}{a_{k,k}} X_\zeta + \frac{b_k -\omega_{\a}}{a_{k,k}} \right) \pmod{\I(\AA)} =\\
     \prod_{\substack {i=1 \\ i\neq k} }^{j+1}\left(1 - X_i^{d_i-1}\right) \prod_{\a=1}^{s}(X_k -\omega_{\a}).
\end{multline*}
Again substituting $X_i =0$ for $ 1 \leq i \leq s_r, i \neq k$, we get 
$$
   \prod_{\a=1}^{s}\left[ X_k- \left( \frac{\omega_{\a}-b_k}{a_{k,k}}\right) \right] =  \prod_{\a=1}^{s}(X_k -\omega_{\a}), 
$$
which holds true if and only if  $\left\{ \frac{\omega_{1}-b_k}{a_{k,k}}, \ldots, \frac{\omega_{s}-b_k}{a_{k,k}} \right\} = \{\omega_1, \ldots, \omega_s \}$. This is equivalent to determining $\sigma_{a_{k,k}, b_k} \in \AGL(1, F_{t_k})$ such that $a_{k,k} \Omega+b_k =\Omega$. Define the set 
\begin{equation}\label{eqn_Delta_Omega}
    \Delta_{\Omega} = \left\{\sigma_{a_{k,k},b_k} \in\AGL(1, F_{t_k}) : \{a_{k,k}\w_1+b_k, \ldots, a_{k,k}\w_s+b_k\} = \{\w_1, \ldots, \w_s \} \right\}.
\end{equation}
Therefore, our equation becomes 
\begin{multline*}
    \prod_{\substack {i=1 \\ i\neq k} }^{j+1}\left(1 - \left( \sum_{\substack {\zeta=1 \\ \zeta \neq k} }^{s_{t_i}} a_{\zeta, k} X_\zeta \right)^{d_i-1}\right) \prod_{\a=1}^{s}\left( \left(\sum_{\substack{\zeta =1\\\zeta \neq k}}^{s_{t_k}} \frac{a_{k,\zeta }}{a_{k,k}} X_\zeta \right) +X_k -\omega_{\a}\right) \pmod{\I_{\AA}} = \\
    \prod_{\substack {i=1 \\ i\neq k} }^{j+1}\left(1 - X_i^{d_i-1}\right) \prod_{\a=1}^{s}(X_k -\omega_{\a}).
\end{multline*}
Moreover, if $A' = (a_{i,j}^{'}) \in \G$ where
$$
    a'_{i,j} = \begin{cases}
        a_{i,j} & \text{ if } i \neq k,\\
        a_{k,j} & \text{ if } i = k \text{ and } j \in \{s_{t_k}+1, \ldots, m \} \cup \{k\}.\\
    \end{cases}
$$ 
Then, it follows from Lemma \ref{lemma_stab_3} that $(\g, \sigma_{A',b}) \in \stab_{\GA}(\c_{h_{k}^{\Omega}})$ and hence, we can write
\begin{multline*}
    \stab_{\GA}(\c_{h_{k}^{\Omega}}) = \{ (\g, \sigma_{A,b}) \in \GA: a_{i,k}= 0 = b_i \text{ for } i \in [j+1]\setminus \{k\}, \g =(a_{k,k})^{-s}, \\ 
    (a_{k,k},b_k) \in \Delta_{\Omega} \text{ and } \widehat{h_k} (\sigma_{\widehat{A_k},0} (\widehat{\X_k})) = \widehat{h_k}(\widehat{\X_k})\}, 
\end{multline*}
where  $\Delta_{\Omega}$ is given by equation \eqref{eqn_Delta_Omega}$, \widehat{\X_k} = (X_1, \ldots, X_{k-1}, X_{k+1}, \ldots, X_m) \in \F_q^{m-1}$ and 
$\widehat{A_k}$ is obtained from $A$ by removing the $k^{\rm th}$-row and $k^{\rm th}$-column of $A$ and $\sigma_{\widehat{A_k},0} \in \Aff(\widehat{\AA_{k}})$
and $\widehat{h_k} \in \Fq[\widehat{\X_k}]$ is given by  
$$
    \widehat{h_k}(\widehat{\X_k}) = \prod_{\substack {i=1 \\ i\neq k} }^{j+1}\left(1 - X_i^{d_i-1}\right).
$$
Note that $\widehat{h_k}(\widehat{\X_k})$ is a polynomial of degree $(u-s)$ and corresponds to a minimum weight codeword of $\AC_{q}(u-s, \widehat{\AA_k})$. Therefore, 
$$
    |\stab_{\GA}(\c_{h_{k}^{\Omega}})| = d_{k}^{s_{t_k}-1} (d_{j+2}\cdots d_m)^2  |\Delta_{\Omega}| |\stab_{\widehat{\G_k}}(\c_{\widehat{h_k}})|.
$$
We know that $\c_{\widehat{h_k}} \in \AC_{q}(u-s, \widehat{\AA_k})$ where $u-s = \sum\limits_{\substack{i =1\\i \neq k}}^{j+1} (d_i -1)$. Therefore, 
$$
    |\N_{q}(u-s, \widehat{\AA_k})| = \frac{|\GG_{\widehat{\AA_k}}|}{|\stab_{\GG_{\widehat{\AA_k}}}(\c_{\widehat{h_k}})|} = (q-1) \left(\prod_{\substack{i =1\\i \neq k}}^{j+1} d_i\right) \stirling{\mu_r}{j+1-s_{r-1}}_{d_{j+1}}.
$$
But equation \eqref{eqn_stab_GA} gives $\left| \stab_{\GG_{\widehat{\AA_k}}}(\c_{\widehat{h_k}}) \right| = \left( \prod\limits_{i=j+2}^{m}d_i \right) \left| \stab_{\widehat{\G_{k}}}(\c_{\widehat{h_k}}) \right|$. Therefore, 

\begin{equation*}
    \begin{split}
       \left| \N_{q}^{(k, \Omega)}(u, \AA) \right| &= \frac{|\GA|}{|\stab_{\GA}(\c_{h_{k}^{\Omega}})|}  = \frac{|\GA|}{d_{k}^{s_{t_k}-1} (d_{j+2}\cdots d_m)^2  |\Delta_{\Omega}| |\stab_{\widehat{\G_k}}(\c_{\widehat{h_{k}}})|}\\
       &= \frac{|\GA| |\N_{q}(u-s, \widehat{\AA_k})|}{|\GG_{\widehat{\AA_k}}| d_{k}^{s_{t_k}-1} (d_{j+2}\cdots d_m)  |\Delta_{\Omega}|   }\\
       &= \frac{(q-1) d_{k}^{s_{t_k}} (d_k^{\mu_{t_k} }-1) (d_{s_{t_k}+1} \cdots d_m) }{ d_{k}^{s_{t_k}-1} (d_{j+2}\cdots d_m) |\Delta_{\Omega}|} \left(\prod\limits_{\substack{i =1\\i \neq k}}^{j+1} d_i\right) \stirling{\mu_r}{j+1-s_{r-1}}_{d_{j+1}}\\
       &= \frac{B_k}{|\Delta_{\Omega}|},
    \end{split}
\end{equation*}
where 
$$
    B_k = (q-1) d_k  (d_k^{\mu_{t_k}}-1) \left(\prod\limits_{i =s_{t_k}+1}^{j+1} d_i\right) \left(\prod\limits_{\substack{i =1\\i \neq k}}^{j+1} d_i\right) \stirling{\mu_r}{j+1-s_{r-1}}_{d_{j+1}}.
$$
Now, observe that $B_{k}$ is independent of choice of $\Omega$ and hence $|\N_{q}^{(k, \Omega)}|$  depends only on the size of  set $ \Delta_{\Omega}$. Since there are $\eta$ distinct orbits $\Omega_1, \ldots, \Omega_{\eta}$, we must have
$$
    \left| \N_{q}^{(k)}(u, \AA)\right| = \sum_{i=1}^{\eta} \left| \N_{q}^{(k, \Omega_i)}(u, \AA)\right| = \sum_{i=1}^{\eta} \frac{B_k}{ |\Delta_{\Omega_i}|}. 
$$
But, equation \eqref{eqn_stab-reciprocal-sum} gives 
$$
     \left| \N_{q}^{(k)}(u, \AA)\right| =  \frac{B_k}{d_k(d_k -1)} \binom{d_k}{s}, 
$$
and therefore,  
\begin{equation}
    \left| \N_{q}^{(k)}(u, \AA)\right|  = (q-1) \bigg(\prod\limits_{\underset{i \neq k}{i=1}}^{j+1} d_i \bigg) \bigg(\prod\limits_{i=s_{t_k}+1}^{j+1} d_i \bigg) \displaystyle{\stirling{\mu_r}{j+1-s_{r-1}}_{d_{j+1}}  \stirling{\mu_{t_k}}{1}_{d_k}}\binom{d_k}{d_{j+1}-\ell}.
\end{equation}
Hence, the proof is complete. 
\end{proof}

\begin{remark}\label{rmk_tk<r_combined}
    The results in Theorem \ref{thm_d_k=d_j+1 -l} and Theorem \ref{thm_d_k>d_j+1 -l} can be combined since they differ only in the factor $\binom{d_k}{d_{j+1-\ell}}$ which is equal to $1$ in Theorem $\ref{thm_d_k=d_j+1 -l}$. 
\end{remark}

\section{Proof of the Main Theorem}\label{section_conclusion}

Recalling notations from Section \ref{section_intro}, 
let $m \geq 1$ and $\Fq$ be the finite field with $q$ elements. Consider a subset $\AA$ of $\F_q^m$ defined by 
$$
\mathcal{A}: = A_1 \times \cdots \times A_m ,
$$
where $A_1, \ldots, A_m$ are subsets of $\Fq$ and let $|A_i| =d_i$ for $1 \leq i \leq m$. Moreover, we assume that $1 < d_1 \leq \cdots \leq d_m \leq q$.  We consider a special case for which $\AA$ is given by
$$
    \AA = F_1 \times \cdots \times F_1 \times \cdots \times F_{\lambda} \times \cdots \times F_{\ld}  = F_1^{\mu_1} \times \cdots \times F_{\ld}^{\mu_{\ld}} \subseteq \F_q^m,
$$
where $F_1\subsetneq \ldots \subsetneq F_{\lambda}$ are distinct  subfields of $\Fq$ and each $F_t$ is repeated $\mu_t$ times in $\AA$ for some positive integers $\mu_1, \ldots, \mu_{\ld}$. Note that $\mu_1 + \cdots+ \mu_{\ld} = m$. Let $s_0 =0 $ and for $1 \leq t \leq \lambda$, define $s_t = \mu_1+ \cdots+\mu_t$.  The affine Cartesian code $\AC_{q}(u, \AA)$ is given by $ \Ev(\S_{\leq u}(\AA))$ where the evaluation map is defined by equation \eqref{eqn_Ev}. We consider affine Cartesian codes for this special case of nested subfields and focus on determining the number of minimum weight codewords of $\AC_{q}(u, \AA)$ for $0 \leq u \leq K$. We denote by $\N_{q}(u, \AA)$, the set of minimum weight codewords of  $\AC_{q}(u, \AA)$.

The next theorem enumerates the minimum weight codewords for this special case of affine Cartesian codes for all values of $u$, i.e., $0 \leq u \leq K$.  

\begin{theorem}\label{thm_main-proof}
Let $m \geq 1$ and  $\AA = F_1 \times \cdots \times F_1 \times \cdots \times F_{\lambda} \times \cdots \times F_{\ld}  = F_1^{\mu_1} \times \cdots \times F_{\ld}^{\mu_{\ld}} \subseteq \F_q^m $ as before. Then,
$$
    \left| \N_{q}(0, \AA) \right| = q-1.
$$
Let $K = \sum\limits_{i=1}^{m} (d_i-1)$. For $1 \leq u \leq K$, write 
$$
    u = \sum_{i=1}^{j}(d_i-1)+\ell,
$$
where $j, \ell$ are uniquely determined integers such that $0 \leq j < m$ and $0 < \ell \leq d_{j+1}-1$. Then
\begin{equation*}
    \left| \N_{q}(u, \AA) \right| = \begin{cases}
        (q-1) \bigg(\prod\limits_{\underset{i \neq k}{i=1}}^{j+1} d_i \bigg)  \displaystyle{\stirling{\mu_r}{j+1-s_{r-1}}_{d_{j+1}}} & \text{ if } \ell = d_{j+1}-1,\\
        \left| \N_{q}^{(j+1)}(u, \AA) \right| +  \sum\limits_{t = t_{k_0}}^{r-1} \left| \N_{q}^{(s_{t})}(u, \AA) \right|   & \text{ if } \ell < d_{j+1}-1,\\
    \end{cases}
\end{equation*}
where $r = t_{j+1}$, $1 \leq k_0 \leq j+1$ is the least integer $k$ such that $d_{k} \geq d_{j+1}-\ell$ and  $|\N_{q}^{(k)}(u, \AA)|$ is given by
    \begin{itemize}
        \item If $t_k = r$, 
        $$
            \left| \N_{q}^{(k)}(u, \AA) \right| = (q-1) \bigg(\prod\limits_{\underset{i\neq k}{i=1}}^{j+1} d_i \bigg) \displaystyle{\stirling{\mu_r}{j-s_{r-1}}_{d_{j+1}}  \stirling{s_{r}-j}{1}_{d_k}} \binom{d_{k}}{d_{j+1}-\ell}. 
        $$

        \item If $t_k < r$, 
        $$
            \left| \N_{q}^{(k)}(u, \AA) \right| =   (q-1) \bigg(\prod\limits_{\underset{i \neq k}{i=1}}^{j+1} d_i \bigg) \bigg(\prod\limits_{i=s_{t_k}+1}^{j+1} d_i \bigg) \displaystyle{\stirling{\mu_r}{j+1-s_{r-1}}_{d_{j+1}}  \stirling{\mu_{t_k}}{1}_{d_k}}\binom{d_k}{d_{j+1}-\ell}.  
        $$
    \end{itemize}
\end{theorem}  

\begin{proof}
If $u=0$, the result follows from equation $\eqref{eqn_N_0}$. For $1 \leq u \leq K$,  we have 
$$
    u = \sum\limits_{i=1}^{j}(d_i-1)+\ell,
$$
for some uniquely determined integers $j, \ell$ such that $0 \leq j <m$ and $0 < \ell \leq d_{j+1}-1$.
 
If $\ell =d_{j+1}-1$, the result directly follows from Theorem \ref{thm_l=d-1}. If $\ell < d_{j+1}-1$, then it follows from Proposition \ref{prop_empty-intersection} and equation \eqref{eqn_N_sum} that 
$$
    |\N_{q}(u, \AA)| =\left| \N_{q}^{(j+1)}(u, \AA) \right| +  \sum\limits_{t = t_{k_0}}^{r-1} \left| \N_{q}^{(s_{t})}(u, \AA) \right|, 
$$
and for $1 \leq k \leq j+1$, $| \N_{q}^{(k)}(u, \AA) |$ is determined separately for different cases. If $t_k=r$, then $| \N_{q}^{(j+1)}(u, \AA) |$ is given by Theorem \ref{thm_tk=r} whereas, if $t_k< r$, the result follows from Theorem \ref{thm_d_k=d_j+1 -l}, \ref{thm_d_k>d_j+1 -l} and Remark \ref{rmk_tk<r_combined}. 
\end{proof}

We have seen in Section \ref{section_intro} that Reed-Solomon codes and Reed-Muller codes are special cases of affine Cartesian codes. Thus, we deduce the number of minimum weight codewords for these codes from our result. 

\subsection{Reed-Solomon codes}
Let $n, k$ be positive integers such that $k \leq n \leq q$. Let $T=\{p_1, \ldots, p_n \} \subseteq \Fq$. Then
$$
    \RS_{q}(u,n) :=\{ (f(p_1), \ldots, f(p_n)): f \in \Fq[X], \deg(f) < k  \}
$$
is a linear code of length $n$ and dimension $k$, called the \textit{Reed-Solomon code}. This is a Maximum Distance Separable (MDS) code and its minimum distance is given by $n-k+1$. Note that $\RS_{q}(k,n) =\Fqn$ if  $k \geq n$. Therefore, it suffices to consider $1 \leq k \leq n$.  

\begin{cor}
For $1 \leq k \leq n$, let $\N(\RS_{q}(k,n))$ denotes the set of minimum weight codewords of Reed-Solomon code $\RS_{q}(k,n)$. Then
$$
    |\N(\RS_{q}(k,n))| = (q-1)\binom{n}{n-k+1}.
$$
\end{cor}
     
\begin{proof}
We know that if $m=1$, then $\RS_{q}(k,n) =\AC_{q}(k-1, \AA)$ where $\AA = T \subseteq \Fq$. If $k =1$, then 
$$
     |\N(\RS_{q}(k,n))| =q-1.
$$
If $2  \leq  k \leq n$, then upon comparing, we have $|\AA| = n =d_1$,  $r=1$, $\mu_r=1$ and hence $j=0$ and therefore $\ell= k-1$. Hence, it follows from Theorem \ref{thm_tk=r} that
\begin{equation*}
    \begin{split}
        |\N(\RS_{q}(k,n))| &=  |\N_{q}(k-1, T)| = (q-1) \stirling{1}{0}_{n} \stirling{1}{1}_{n} \binom{n}{n-(k-1)}\\
        &= (q-1) \binom{n}{n-k+1}.
    \end{split}
\end{equation*}
\end{proof}

\subsection{Reed-Muller codes}
Let $m$ be a positive integer and let $u \geq 0$. Write $\Fqm = \{P_1, \ldots, P_{q^m}\}$. Then
$$
    \RM_{q}(u,m): = \{(f(P_1), \ldots, f(P_{q^m}): f \in \Fq[X_1, \ldots, X_m], \deg(f) \leq u \}
$$
is a linear code of length $q^m$, called \textit{the Reed-Muller code of order $u$}. We have already seen that if $\AA =\Fqm$, then $\AC_{q}(u,m) = \RM_{q}(u,m)$. Thus, $\RM_{q}(u,m) =\F_q^{q^m}$ if $u \geq m(q-1)$. Therefore, we consider $0 \leq u \leq m(q-1)$. 

\begin{cor}
For $0 \leq u \leq m(q-1)$, write $u = t(q-1)+s $ for unique integers $t,s$ such that $0 \leq t \leq m$ and $0 \leq s< q-1$. Let $\mathrm{N}_{q}(u,m)$ denotes the set of minimum weight codewords of $\RM_{q}(u,m)$, then 
$$
    \left| \mathrm{N}_{q}(u,m) \right| =
    \begin{cases}
        (q-1)q^{t} \stirling{m}{t}_{q} & {\text{ if }s =0},\\[.5em]    
        (q-1)q^{t} \stirling{m}{t}_{q}  \stirling{m-t}{1}_{q} \binom{q}{s}  & {\text{ if } s \neq 0}.
    \end{cases}
$$
\end{cor}

\begin{proof}
Upon comparing with $\AC_{q}(u, \AA)$, we have $d_1 = \cdots=d_m =q$, $r=1$, $\mu_r=s_r=m$.  If $s=0$, then $u=t(q-1)$, and hence, we must have  $\ell= d_{j+1}-1$ and therefore $t= j+1$ when we write $u = \sum\limits_{i=1}^{j}(d_i-1)+\ell$. Thus, it follows from Theorem \ref{thm_l=d-1} that 
$$
    |\mathrm{N}_{q}(u,m)| = (q-1)(d_1 \cdots d_{j+1}) \stirling{\mu_r}{j+1}_{d_{j+1}} = (q-1) q^t \stirling{m}{t}_{q}. 
$$

If $s \neq 0$, we have $t=j$ and $s =\ell$, and hence, Theorem \ref{thm_tk=r} gives
$$
    |\mathrm{N}_{q}(u,m)| = (q-1)(d_1 \cdots d_{j}) \stirling{\mu_r}{j}_{d_{j+1}} \stirling{s_r-j}{1}_{d_{j+1}} \binom{d_{j+1}}{\ell} = (q-1) q^t \stirling{m}{t}_{q} \stirling{m-t}{1}_{q} \binom{q}{s}.
$$
\end{proof}

\subsection{Examples}
We illustrate our results with the help of two examples. Consider $m=3, \; q=4$, and let 
$$
    \AA = \F_2 \times \F_2 \times \F_4 \subseteq \F_4^{3}.
$$
Therefore, in this case we have $\mu_1=2, \;  \mu_2=1$ and $d_1=d_2=2, \; d_3=4$. Thus, we must have $0 \leq u \leq K = 5$. The table below gives $|\N_{4}(u, \AA)|$ for $1 \leq u \leq 5$. 

\medskip


\begin{center}
\begin{tabular}
{|c|c|c|c|c|c|c|c|c|c|c|c|}
\hline
 ${u}$ &  ${(j,\ell)}$  &   \hspace{0.1cm} ${k_0}$ \hspace{0.1cm} & \hspace{0.1cm} ${k}$  \hspace{0.1cm} & \hspace{0.1cm} ${|C|}$ \hspace{0.1cm} &${\delta_{q}(u, \AA)}$ &  $|{\N_{q}^{(k)}(u,\AA)}|$ &   $|{\N_{q}(u,\AA)}|$ & \hspace{0.1cm} {Ref.}  \hspace{0.1cm}\\[.1em]
\hline
$1$ & $(0,1)$ &$1$ & $1$ & $4^{4}$ & $8$ & $18$ & $18$ & Thm. \ref{thm_l=d-1}\\[1em]
$2$ & $(1,1)$ & $2$ & $2$ &  $4^{8}$ & $4$ & $12$ & $12$ & Thm. \ref{thm_l=d-1}\\[1em]
$3$ & $(2,1)$ & $3$ & $3$  & $4^{12}$ & $3$ & $48$ & $48$ & Thm. \ref{thm_tk=r}\\[1em]
$\multirow{2}{*}{4}$ & $\multirow{2}{*}{(2,2)}$ &  $\multirow{2}{*}{1}$ & $1,2$ &  $\multirow{2}{*}{$4^{15}$}$ & $\multirow{2}{*}{2}$ & $288$ & $\multirow{2}{*}{\color{black}360}$& Thm. \ref{thm_d_k=d_j+1 -l} \\
 &   & & $3$ &  & & $72$ & & Thm. \ref{thm_tk=r}\\[1em]
$5$ &  $(2,3)$ & $3$ & $3$ & $4^{16}$ & $1$ & $48$ & $48$& Thm. \ref{thm_l=d-1}\\
\hline
\end{tabular}
\end{center}

\bigskip
  
In the next example, we consider $m=2$,  $q=9$,  and let $\AA = \F_3 \times \F_9 \subseteq \F_9^2$.  The table below illustrates all values of $|\N_{9}(u, \AA)|$ for $1 \leq u \leq K =10 $. 

\bigskip

\begin{center}
\begin{tabular}
{|c|c|c|c|c|c|c|c|c|c|c|c|}
\hline
 ${u}$ &  ${(j,\ell)}$  &   \hspace{0.1cm} ${k_0}$ \hspace{0.1cm} & \hspace{0.1cm} ${k}$  \hspace{0.1cm} & \hspace{0.1cm} ${|C|}$ \hspace{0.1cm} &${\delta_{q}(u, \AA)}$ &  $|{\N_{q}^{(k)}(u,\AA)}|$ &   $|{\N_{q}(u,\AA)}|$ & \hspace{0.1cm} {Ref.}  \hspace{0.1cm}\\[.1em]
\hline
$1$ & $(0,1)$ &$1$ & $1$ & $9^{3}$ & $18$ & $24$ & $24$ & Thm. \ref{thm_tk=r}\\[1em]
$2$ & $(0,2)$ & $1$ & $1$ &  $9^{6}$ & $9$ & $24$ & $24$ & Thm. \ref{thm_l=d-1}\\[1em]
$3$ & $(1,1)$ & $2$ & $2$  & $9^{9}$ & $8$ & $216$ & $216$ & Thm. \ref{thm_tk=r}\\[1em]
$4$ & $(1,2)$ &  $2$ & $2$ &  $9^{12}$ & $7$ & $864$ & $864$& Thm. \ref{thm_tk=r} \\[1em]
$5$ &  $(1,3)$ & $2$ & $2$ & $9^{15}$ & $6$ & $2016$ & $2016$& Thm. \ref{thm_tk=r}\\[1em]
$6$ &  $(1,4)$ & $2$ & $2$ & $9^{18}$ & $5$ & $3024$ & $3024$& Thm. \ref{thm_tk=r}\\[1em]
$7$ &  $(1,5)$ & $2$ & $2$ & $9^{21}$ & $4$ & $3024$ & $3024$& Thm. \ref{thm_tk=r}\\[1em]
$\multirow{2}{*}{8}$ &  $\multirow{2}{*}{(1,6)}$ & $\multirow{2}{*}{1}$ & $1$ & $\multirow{2}{*}{$9^{24}$}$ & $\multirow{2}{*}{3}$ & ${\color{black}648}$ & $\multirow{2}{*}{\color{black} 2664}$& Thm. \ref{thm_d_k=d_j+1 -l}\\
 &  &  & $2$ &  &  & ${\color{black} 2016}$ & & Thm. \ref{thm_tk=r}\\[1em]

 $\multirow{2}{*}{9}$ &  $\multirow{2}{*}{(1,7)}$ & $\multirow{2}{*}{1}$ & $1$ & $\multirow{2}{*}{$9^{26}$}$ & $\multirow{2}{*}{2}$ & ${\color{black}1944}$ & $\multirow{2}{*}{\color{black} 2808}$ & Thm. \ref{thm_d_k>d_j+1 -l}\\
 &  &  & $2$ &  &  & ${\color{black} 864}$ & & Thm. \ref{thm_tk=r}\\[1em]
$10$ &  $(1,8)$ & $2$ & $2$ & $9^{27}$ & $1$ & $216$ & $216$& Thm. \ref{thm_l=d-1}\\
\hline
\end{tabular}
\end{center}

\end{document}